\documentclass{lmcs} 

\usepackage{hyperref}

\usepackage{times}
\usepackage{amssymb}
\usepackage{verbatim}
\usepackage{graphicx}
\usepackage{color}
\usepackage{amsthm}
\usepackage{xcolor,pgf,tikz,pgflibraryarrows,pgffor,pgflibrarysnakes}
\usepackage{amsmath}
\usepackage{amsfonts}
\usepackage{amssymb}
\usepackage{tikz}

\usetikzlibrary{backgrounds} 
\usetikzlibrary{patterns}
{\bfseries}{\rmfamily} 

\newtheorem{theorem}{Theorem}[section]

\newtheorem{proposition}{Proposition}[theorem]
\theoremstyle{remark}

\theoremstyle{remark}

\newcommand{\false}{\mathsf{False}}

\newcommand{\Ss}{\mathcal{S}}

\newcommand{\N}{\mathbb N}

\newcommand{\wB}{\Box^{\mathsf{ns}}}
\newcommand{\wU}{\until^{\mathsf{ns}}}
\newcommand{\wF}{\fut^{\mathsf{ns}}}

\newcommand{\wS}{\since^{\mathsf{ns}}}

\newcommand{\until}{\:\mathsf{U}}
\newcommand{\since}{\:\mathsf{S}}

\newcommand{\R}{\:\mathbb{R}}
\newcommand{\fut}{\Diamond}
\newcommand{\past}{\mbox{$\Diamond\hspace{-0.29cm}-$}}

\mathchardef\mhyphen="2D
\mathchardef\mhyph="2D

\newcommand{\nex}{\mathsf{O}}
\newcommand{\nx}{\mathsf{O}}\newcommand{\prev}{\mathsf{\bar{O}}}

\newcommand{\sbf}{\Box}
\newcommand{\sbm}{\boxminus}

\newcommand{\tap}{\mathsf{T}}

\newcommand{\mtlpwuisi}{\mathsf{MTL}[\until_I, \since_I]}

\newcommand{\mtl}{\mathsf{MTL}}
\newcommand{\tptl}{\mathsf{TPTL}}

\newcommand{\mtlpwuisnp}{\mathsf{MTL}[\until_I,\since_{np}]}
\newcommand{\mtlpwunpsi}{\mathsf{MTL}[\until_{np},\since_I]}

\newcommand{\intintervaln}{\mathcal{I}_\mathsf{nat}}
\newcommand{\intinterval}{\mathcal{I}_\mathsf{int}}
\newcommand{\ltl}{\mathsf{LTL}}

\newcommand{\mitl}{\mathsf{MITL}}

\newcommand{\oomit}[1]{}

\newcommand{\optptl}{\mbox{$\mathsf{OpTPTL}^1$}}
\tikzstyle{block} = [rectangle, draw = white, fill=white!20, 
text width=6em, text centered, rounded corners, minimum height=2em]
\tikzstyle{bigblock} = [rectangle, draw = white, fill=white!20, 
text width=8em, text centered, rounded corners, minimum height=3em]
\tikzstyle{line} = [draw, -latex']
\tikzstyle{cloud} = [rectangle, draw = white, fill=white!20, 
text width=11em, text centered, rounded corners, minimum height=3em]

\newcommand{\regm}{\mathsf{Rat}}
\newcommand{\fregm}{\mathsf{FRat}}
\newcommand{\freg}{\mathsf{FRat}}
\newcommand{\pregm}{\mathsf{PRat}}
\newcommand{\natptl}{\mathsf{NA\text{-}1\text{-}TPTL}}
\newcommand{\pnatptl}{\mathsf{NA^+\text{-}1\text{-}TPTL}}
\newcommand{\nnatptl}{\mathsf{NA^-\text{-}1\text{-}TPTL}}
\newcommand{\patptl}{\mathsf{PA\text{-}1\text{-}TPTL}}
\newcommand{\regmtl}{\mathsf{RatMTL}}
\newcommand{\re}{\mathsf{re}}
\newcommand{\emitl}{\mathsf{EMITL}}

\newcommand{\pnregmtl}{\mathsf{PnEMTL}}

\newcommand{\fregmk}{\mathcal{F}}
\newcommand{\pregmk}{\mathcal{P}}
\newcommand{\pnaem}{\pnregmtl[\fregmk_{na}, \pregmk_I]}
\newcommand{\nnaem}{\pnregmtl[\fregmk_I, \pregmk_{na}]}

\newcommand{\cal}[1]{\mathcal{#1}}
\newcommand{\iecm}{\mathsf{IECM}}
\newcommand{\Cc}{\mathsf{C}}
\newcommand{\Pp}{\mathsf{P}}

\usepackage{hyperref}
\usepackage[normalem]{ulem}
\usepackage{stmaryrd}
\usepackage{amssymb}
\usepackage{empheq}
\usepackage{graphicx}
\usepackage{xcolor,pgf,tikz,pgflibraryarrows,pgffor,pgflibrarysnakes}
\usepackage{microtype}
\usepackage{tikz}

\usetikzlibrary{fit} 
\usetikzlibrary{backgrounds} 

\usetikzlibrary{calc}
\tikzstyle{RectObject}=[rectangle,fill=white,draw,line width=0.5mm]
\tikzstyle{line}=[draw]
\tikzstyle{arrow}=[draw, -latex] 
\usepgflibrary{shapes}
\usetikzlibrary{snakes,automata}
\usetikzlibrary{calc}

\newcommand{\true}{\mathsf{True}}

\usepackage{physics}  
\usepackage{mathtools}
\usepackage{amsmath}
\usepackage{mathrsfs}
\usepackage{amssymb}
\usepackage{times}
\usepackage{todonotes}
\usepackage{latexsym}
\usepackage{amssymb,amsmath,amsthm}
\usepackage{gastex,pstricks}
\keywords{Real-Time Logics, Metric Temporal Logic, Timed Propositional Temporal Logic, Satisfiability checking, Decidability, Non-Punctuality, Non-Adjacency, Expressiveness}
\usepackage{hyperref}
\theoremstyle{plain} 

\begin{document}

\title{Openness and Partial Adjacency in One Variable TPTL.}

\author[S.N.~Krishna]{Shankara Narayan Krishna}[a]
\author[K.~Madnani]{Khushraj Madnani}[b]
\author[A.Nag]{Agnipratim Nag}[a]
\author[P.K.Pandya]{Paritosh Pandya}[a]

\address{Indian Institute of Technology Bombay, India}	
\email{krishnas@cse.iitb.ac.in, agnipratim.nag@iitb.ac.in, pandya@tifr.res.in}  

\address{Max Planck Institute for Software Systems, Germany}	
\email{kmadnani@mpi-sws.org}  





\begin{abstract}
Metric Temporal Logic (MTL) and Timed Propositional Temporal Logic (TPTL) are prominent real-time extensions of Linear Temporal Logic (LTL). MTL extends LTL modalities, Until, $\until$ and Since, $\since$ to family of modalities $\until_I$ and $\since_I$, respectively, where $I$ is an interval of the form $\langle l,u \rangle$ to express real-time constraints. On the contrary, TPTL extends LTL by real-valued freeze quantification, and constraints over those freeze variables to do the same. It is well known that one variable fragment of TPTL is strictly more expressive than MTL. In general, the satisfiability checking problem for both MTL and TPTL is undecidable. 
MTL enjoys the benefits of relaxing punctuality. That is, satisfiability checking for Metric Interval Temporal Logic (MITL), a subclass of MTL where the intervals are restricted to be of the form $\langle l, u \rangle$ where $l<u$, is decidable with elementary complexity (EXPSPACE complete). Moreover, Partially Punctual Metric Temporal Logic (PMTL), a subclass of MTL where punctual intervals are only allowed in either $\until$ modalities or $\since$ modalities, but not both, is also decidable over finite timed words with non-primitive recursive complexity.

In case of TPTL, punctuality can be trivially recovered due to freeze quantifiers and boolean over guards. Hence, we study a more restrictive version of non-punctuality, called Openness. Intuitively, this restriction only allows a property to be specified within timing intervals which are topologically open. We show that even with this restriction, 1-TPTL is undecidable. Our results make a case for a the new refined notion of non-adjacency by Krishna et. al. for getting a decidable fragment of 1-TPTL, called non-adjacency. We extend the notion of non-adjacency to partial adjacency, where the restriction is only applicable in either past or future but not in both directions. We show that partially adjacent 1- TPTL (PA-1-TPTL) is decidable over finite timed words. Moreover, it is strictly more expressive than PMTL, making it the most expressive boolean closed decidable timed logic known in the literature.

\end{abstract}

\maketitle

\section{Introduction}
Metric Temporal Logic $\mtlpwuisi$ is a well established logic useful for specifying quantitative properties of real-time systems. The main modalities of  $\mathsf{MTL}$ are $\until_I$ (read ``until $I$'') and 
$\since_I$ (read ``since $I$''), where $I$ is a time interval with endpoints in $\mathbb{N}$. These 
formulae are interpreted over timed behaviours or timed words.  For example, a formula $a \until_{[2,3]} b$ is satisfied by a position $i$ of a timed word $\rho$ if and only if there is a position $j$ 
strictly in the future of $i$ where $b$ is true, and 
at all intermediate positions between $i$ and $j$, $a$ is true; moreover, the difference 
in the timestamps of $i$ and $j$ must lie in the interval $[2,3]$.   Similarly, $a \since_{[2,3]} b$ is true at a point $i$ if and only if there is a position $j$ strictly 
in the past of $i$ where $b$ is true, and at all intermediate positions between $i$ and $j$, $a$ is true; 
further, the difference in the timestamps between $i$ and $j$ lie in the interval $[2,3]$. 

 
 In their seminal paper, Alur and Henzinger \cite{AH93} showed that the satisfiability of full  $\mathsf{MTL}$, with until and since modalities
 is  undecidable even over finite words. This ability to encode undecidable problems is due to the presence of punctual intervals, i.e., intervals of the form $[x,x]$. This allows the logic to specify constraints like ``an event $a$ occurs exactly after 5 time units, $\top \until_{[5,5]} a$.'' In practice, such exact constraints are not used extensively. Hence, Alur et al. studied the non-punctual fragment of MTL called Metric Interval Temporal Logic (MITL) in  \cite{AlurFH91} \cite {AFH96} where  the time intervals used in the until, since modalities are non-punctual, i.e. of the form $\langle x, y\rangle$ where $x < y$. They show that the satisfiability becomes decidable over finite as well as infinite timed words with EXPSPACE complexity.

 The satisfiability of the future only fragment of $\mathsf{MTL}$, where since modalities are not used (MTL[$\until_I$]), was open for a long time. Ouaknine and Worrell \cite{Ouaknine05} showed its decidability via a reduction to 1-clock Alternating Timed Automata over finite timed words.  
  A natural extension to both these problems studied in \cite{AFH96}\cite{AlurFH91} \cite{Ouaknine05} is to ask what happens to the decidability and expressiveness of $\mtlpwuisnp$, subclass of MTL where $\since$ modalities are non-punctual, when interpreted over finite timed words. This was resolved by Krishna et. al. in \cite{time14}.

Timed Propositional Temporal Logic (TPTL) extends LTL with freeze quantifiers. A freeze quantifier \cite{AlurH92}\cite{AH94} has the form $x.\varphi$
 with freeze variable $x$ (also called a clock \cite{patricia}\cite{pandyasimoni}).
When it is evaluated at a point $i$ on a timed word, the time stamp of $i$ (say $\tau_i$) is frozen or registered in $x$, and the formula $\varphi$ is evaluated using this value for $x$. Variable $x$ is used in $\varphi$ in a constraint of the form $T-x \in I$; this constraint, when evaluated at a point $j$, checks if $\tau_j -\tau_i \in I$, where $\tau_j$ is the time stamp at point $j$. Here $T$ can be seen as a special variable giving the timestamp of the present point. For example, the formula  $\varphi = \fut x.(a\wedge \fut (b \wedge T-x \in [1,2] \wedge \fut (c \wedge T-x \in [1,2])))$ asserts that there is a point $i$ in the future where $a$ holds and in its future there is a $b$ within interval $[1,2]$ followed by a $c$ within interval $[1,2]$ from $i$.

The contributions of this paper are two fold:

\begin{enumerate}
\item We study the satisfiability checking problem for a restricted fragment of the TPTL that can only specify properties within topologically open timing intervals. We call this fragment as Open TPTL (denoted by $\optptl$). Notice that the restriction of openness is more restrictive than that of non-punctuality, and with openness punctual guards can not be simulated even with the use of freeze quantifiers and boolean operators. In spite of such a restriction, we show that satisfiability checking of $\optptl$ is as hard as that of 1-$\tptl$ (i.e. undecidable on infinite words, and decidable with non-primitive recursive lower bound for finite words). This implies that it is something more subtle than just the presence of punctual guards that makes the satisfiability checking problem hard for 1-$\tptl$. This makes a strong case for studying non-adjacency of $\tptl$.

\item We define the notion of partial adjacency, generalizing the notion of non-adjacency. Here, we allow adjacent guards but only in one direction. We show that Partially Punctual One variable Timed Propositional Temporal Logic  ($\patptl$) is decidable over finite timed words. Moreover, this logic is the most expressive known decidable subclass of the logic known till date.

\end{enumerate}

\section{Preliminaries}
\label{sec:prelim}
Let $\Sigma$ be a finite set of propositions, and let  $\Gamma = 2^{\Sigma} \setminus \emptyset$ \footnote{ We exclude this empty-set for technical reasons. This simplifies definitions related to equisatisfiable modulo oversampled projections \cite{khushraj-thesis}. Note that this doesn't affect the expressiveness of the models as one can add a special symbol denoting the empty-set.}. 
A finite word over $\Sigma$ is a finite sequence $\sigma = \sigma_1 \sigma_2 \ldots \sigma_n$, where $\sigma_i \in \Gamma$. A finite timed word $\rho$ over $\Sigma$ is a non-empty finite sequence $\rho = (\sigma_1, \tau_1) \ldots (\sigma_n, \tau_n)$ of pairs $(\sigma_i, \tau_i) \in (\Gamma \times \R_{\geq 0})$;  
where $\tau_1=0$ and $\tau_i \leq \tau_j$ for all $1 \leq i \leq j \leq n$ and $n$ is the length of $\rho$ (also denoted by $|\rho|$).  The $\tau_i$ are called time stamps. 
For a timed or untimed word $\rho$, let $dom(\rho) = \{i | 1 \le i \le |\rho|\}$, and 
$\sigma[i]$ denotes the symbol at position $i \in dom(\rho)$.
The set of timed words over $\Sigma$ is denoted $T\Sigma^*$.
Given a (timed) word $\rho$ and $i \in dom(\rho)$, a pointed (timed) word is the pair $\rho, i$. 
Let $\intinterval$ ($\intintervaln$)  be the set of open, half-open or closed time intervals, such that the end points of these intervals are in $\mathbb{Z}\cup \{-\infty,\infty\}$ ($\mathbb{N} \cup \{0,\infty\}$, respectively).

In literature, both finite and infinite words and timed words are studied. As we restrict our study to finite timed words, we call finite timed words and finite words as simply timed words and words, respectively.

\subsection{Linear Temporal Logic}
\label{ltl}
Formulae of $\ltl$ are built over a finite set of propositions $\Sigma$ using Boolean connectives and temporal modalities ($\until$ and $\since$) as follows: \\$\varphi::=a~|\top~|~\varphi \wedge \varphi~|~\neg \varphi~|
~\varphi \until \varphi ~|~ \varphi \since \varphi$,\\ where  $a \in \Sigma$. 
The satisfaction of an $\ltl$ formula is evaluated over pointed words. For a word $\sigma = \sigma_1 \sigma_2 \ldots \sigma_n \in \Sigma^*$ and a point $i \in dom(\sigma)$, the satisfaction of an $\ltl$ formula $\varphi$ at point $i$ in $\sigma$ is defined, recursively, as follows:
\\(i) $\sigma, i \models a$  iff  $a \in \sigma_{i}$, (ii) $\sigma,i  \models \top$ iff  $i \in dom(\rho)$, (iii) $\sigma,i  \models \neg \varphi$ iff  $\sigma,i \nvDash  \varphi$
\\(iv) $\rho,i \models \varphi_{1} \wedge \varphi_{2}$  iff  
$\sigma,i \models \varphi_{1}$ 
and $\sigma,i\ \models\ \varphi_{2}$, (v) $\rho,i \models \varphi_{1} \vee \varphi_{2}$  iff  
$\sigma,i \models \varphi_{1}$ 
or $\sigma,i\ \models\ \varphi_{2}$,
\\(vi) $\sigma,i\ \models\ \varphi_{1} \until \varphi_{2}$  iff  $\exists j > i$, 
$\sigma,j\ \models\ \varphi_{2}$, and  $\sigma,k\ \models\ \varphi_{1}$ $\forall$ $i< k <j$,\\
(vii) $\sigma,i\ \models\ \varphi_{1} \since \varphi_{2}$  iff  $\exists j < i$, 
$\sigma,j\ \models\ \varphi_{2}$, and  $\sigma,k\ \models\ \varphi_{1}$ $\forall$ $j< k <i$.

Derived operators can be defined as follows: $\fut \varphi = \top  \until \varphi$,  and $\sbf \varphi = \neg \fut \neg \varphi$. Symmetrically,  $\past \varphi = \top \since \varphi$,  and $\sbm \varphi = \neg \past \neg \varphi$. 
An LTL formula is said to be in negation normal form if it is constructed out of basic and derived operators above, but
where negation appears only in front of propositional letters. It is well known that every LTL formula can be converted to an equivalent formula which is in negation normal form.

\subsection{Metric Temporal Logic}
$\mathsf{MTL}$ \cite{koymans} is a real time extension of $\mathsf{LTL}$, where timing constraints can be given in terms of intervals along with the modalities.
\subsubsection{Syntax}

Given $\Sigma$,  the formulae of $\mathsf{MTL}$ are built from atomic propositions in $\Sigma$  using boolean connectives and 
time constrained versions of the modalities $\until$ and  $\since$ as follows: 

$$\varphi::=\mathsf{atom} ~|\true~|\varphi \wedge \varphi~|~\neg \varphi~|
~\varphi \until_I \varphi~|~\varphi \since_I \varphi$$

where $\mathsf{atom} \in \Sigma$ and $\Sigma$ is a set of atomic propositions,  $I$ is an open, half-open or closed interval with endpoints in $\N \cup \{0, \infty\}$. 

Formulae of $\mathsf{MTL}$ are interpreted on timed words over a chosen set of propositions. 
Let $\varphi$ be an $\mathsf{MTL}$ formula. Symbolic timed words are more natural models for logics as compared to letter timed words. Thus for a given set of propositions $\Sigma$, when we say that $\varphi$ is  interpreted over $\Sigma$ we mean that $\varphi$ is evaluated over timed words in $\tap \Sigma^*$ unless specifically mentioned otherwise.

\subsubsection{Semantics}
Given a timed  word $\rho$, and an $\mathsf{MTL}$ formula $\varphi$, in the pointwise semantics, the temporal connectives of $\varphi$ quantify over a finite set of positions in $\rho$.  For a timed word $\rho=(\sigma, \tau) \in \tap\Sigma^*$, a position 
$i \in dom(\rho)$, and an $\mathsf{MTL}$ formula $\varphi$, the satisfaction of $\varphi$ at a position $i$ 
of $\rho$ is denoted $(\rho, i) \models \varphi$, and is defined as follows:

\fbox{
\parbox{0.9\linewidth}{
\noindent 
$\rho,i \models \true$ $\leftrightarrow$ $i \in dom(\rho)$\\
$\rho, i \models a$  $\leftrightarrow$  $a \in \sigma_{i}$\\
$\rho,i  \models \neg \varphi$ $\leftrightarrow$  $\rho,i \nvDash  \varphi$\\
$\rho,i \models \varphi_{1} \wedge \varphi_{2}$   $\leftrightarrow$ 
$\rho,i \models \varphi_{1}$ 
and $\rho,i\ \models\ \varphi_{2}$\\
$\rho,i\ \models\ \varphi_{1} \until_{I} \varphi_{2}$  $\leftrightarrow$  $\exists j > i$, 
$\rho,j\ \models\ \varphi_{2}, t_{j} - t_{i} \in I$, and  $\rho,k\ \models\ \varphi_{1}$ $\forall$ $i< k <j$\\
$\rho,i\ \models\ \varphi_{1} \since_{I} \varphi_{2}$  $\leftrightarrow$  $\exists\ j < i$,  
 $\rho,j\ \models\ \varphi_{2}$, $t_{i} - t_{j} \in I$,  
 and  $\rho,k\ \models\ \varphi_{1}$ $\forall$  $j<k < i$
}}

\noindent $\rho$ satisfies $\varphi$ denoted $\rho \models \varphi$ 
if and only if $\rho,1 \models \varphi$. Let $\varphi$ be interpreted over $\tap \Sigma^*$. Then $L(\varphi)=\{\rho \mid \rho, 1 \models \varphi \wedge \rho \in \tap \Sigma^*\}$.
Two formulae $\varphi$ and $\phi$ are said to be language equivalent denoted as $\varphi \equiv \phi$ if and only if $L(\varphi) = L(\phi)$.
Additional temporal connectives are defined in the standard way: 
we have the constrained future and past eventuality operators $\fut_I a \equiv \true \until_I a$ and 
$\past_I a \equiv \true \since_I a$, and their dual  
$\Box_I a \equiv \neg \fut_I \neg a$, 
$\boxminus_I a \equiv \neg \past_I \neg a$. We can also define next operator and previous operators as $\nex_I \phi \equiv \false \until_I  \phi$ and $\prev_I \phi \equiv \false \since_I \phi$, respectively.
Non-strict versions of  operators 
are defined as  $\wF a=a \vee \fut a, 
\wB a\equiv a \wedge \Box a$, $a \wU_I b\equiv b \vee [a \wedge (a \until_I b)]$, $a \wS_I b \equiv b \vee [a \wedge (a \since_I b)]$.
We denote by $\mathsf{MTL}[\until_I,\since_I]$, the class of $\mathsf{MTL}$ formulae with until and since modalities.

\subsection{Natural Subclasses of $\mathsf{MTL}$}
\subsubsection{Metric Interval Temporal Logic ($\mathsf{MITL}$)}
A non-punctual time interval has the form  $\langle a, b \rangle$ with $a \neq b$. 
$\mathsf{MITL}$  is a subclass of $\mathsf{MTL}$ in which all the intervals in the $\until_I, \since_I$ modalities are non-punctual.
We denote non-punctual until, since modalities as $\until_{np}$ and 
$\since_{np}$ respectively, where $np$ stands for {\it non-punctual}.  The syntax of $\mathsf{MITL}$ is as follows:
 
 $$\varphi::=\mathsf{atom} ~|\true~|\varphi \wedge \varphi~|~\neg \varphi~|
 ~\varphi \until_{np} \varphi~|~\varphi \since_{np} \varphi$$
 \begin{theorem}[ \cite{AFH96} ]
	Satisfiability checking for $\mathsf{MITL}$ is decidable over both finite and infinite timed words with EXPSPACE-complete complexity.
\end{theorem}



\subsubsection{The Until Only Fragment of $\mathsf{MTL}$ ($\mathsf{MTL}[\until_I]$)} 
As the name suggests, this subclass of $\mathsf{MTL}$ allows only future operators. The syntax of the future only fragment, denoted by $\mathsf{MTL}[\until_I]$ is defined as
$$\varphi::=\mathsf{atom} ~|\true~|\varphi \wedge \varphi~|~\neg \varphi~|
 ~\varphi \until_I \varphi$$
\begin{theorem}[\cite{Ouaknine05} \cite{OuaknineW06}]
	Satisfiability checking of $\mathsf{MTL}[\until_I]$ is decidable over finite timed words with non primitive recursive complexity and undecidable over infinite timed words. 
\end{theorem}

\subsubsection{Partially Punctual Metric Temporal Logic, PMTL}
Let $\mtlpwuisnp$ be a subclass of $\mtl$ with the following grammar.
 $$\varphi::=\mathsf{atom} ~|\true~|\varphi \wedge \varphi~|~\neg \varphi~|
 ~\varphi \until_{I} \varphi~|~\varphi \since_{np} \varphi$$

Similarly, $\mtlpwunpsi$ is a subclass of $\mtl$ with the following grammar.
 $$\varphi::=\mathsf{atom} ~|\true~|\varphi \wedge \varphi~|~\neg \varphi~|
 ~\varphi \until_{np} \varphi~|~\varphi \since_{I} \varphi$$

PMTL is a subclass of MTL containing formulae $\varphi \in \mtlpwuisnp \cup \mtlpwunpsi$. In other words, punctual guards are only allowed in either $\until$ modalities or $\since$ modalities but not both.

\subsection{Timed Propositional Temporal Logic}
Another prominent real time extension of linear temporal logic is $\mathsf{TPTL}$ \cite{AH94}, where timestamps of the points of interest are registered in a set of real-valued variables called freeze variables. The timing constraints are then specified using constraints on these freeze variables as shown below. 
 
 The syntax of $\mathsf{TPTL}$ is defined by the following grammar:
 $$\varphi::=\mathsf{atom}~|\true~|\varphi \wedge \varphi~|~\neg \varphi~|
 ~\varphi \until \varphi~|~\varphi \since \varphi~|~y.\varphi~|~T-y\in I$$
 where $y$ is a freeze variable, $I \in \intinterval$.
 
$\mathsf{TPTL}$ is interpreted over finite timed words over $\Sigma$. For a timed word $\rho=(\sigma_1,t_1)\dots(\sigma_n,t_n)$, we define the satisfiability relation, $\rho, i, \nu \models \phi$ saying that the formula $\phi$ is true at position
 $i$ of the timed word $\rho$ with valuation $\nu$ over all the freeze variables.
 
 \fbox{
 	\begin{tabular}{l c l}
 		$\rho,i, \nu \models \true$ & $\leftrightarrow$ & $i \in dom(\rho)$\\
 		$\rho, i, \nu \models a$  & \; $\leftrightarrow$ & \; $a \in \sigma_{i}$\\
 		$\rho,i,\nu  \models \neg \varphi$ & \; $\leftrightarrow$  & \; $\rho,i,\nu \nvDash  \varphi$\\
 		$\rho,i,\nu \models \varphi_{1} \wedge \varphi_{2}$  & \; $\leftrightarrow$  & \; $\rho,i,\nu \models \varphi_{1}$ 
 		and $\rho,i,\nu\ \models\ \varphi_{2}$\\
 		$\rho,i,\nu \models x.\varphi $  & \; $\leftrightarrow$ & \; $\rho,i,\nu[x \leftarrow t_i] \models \varphi$\\
 		$\rho,i,\nu \models T-x \in I $  & \; $\leftrightarrow$ & \; $t_i - x \in I$\\
 		$\rho,i,\nu\ \models\ \varphi_{1} \until \varphi_{2}$  & \; $\leftrightarrow$  & \; $\exists j > i$, 
 		$\rho,j,\nu \ \models\ \varphi_{2}$,  and  $\rho,k,\nu \ \models\ \varphi_{1}$ $\forall$ $i < k < j$\\
 		$\rho,i,\nu\ \models\ \varphi_{1} \since \varphi_{2}$  & \; $\leftrightarrow$  & \; $\exists\ j < i$, 
 		$\rho,j, \nu\ \models\ \varphi_{2}$,  and  $\rho,k,\nu\ \models\ \varphi_{1}$ $\forall$  $j < k < i$
 	\end{tabular}\\
 }
\\
\textbf{Remark} In the original paper, $\tptl$ was introduced with {\it non-strict} Until ($\until^{ns}$) and Since ($\since^{ns}$)  modalities along with next($\nex$) and previous ($\prev$) operators as non-strict until (and since) are not able to express next (and previous, respectively). While $\false \until \varphi \equiv \nex (\varphi)$ and $\varphi_1 \until^{ns} \varphi_2 \equiv \varphi_2 \vee (\varphi_1 \wedge \varphi_1 \until \varphi_2)$. Similar, identity holds for since. Thus, the use of non-strict modalities along with next and previous modalities can be replaced with corresponding strict modalities. To maintain uniform notations within the thesis, we choose to use strict modalities. $\rho$ satisfies $\phi$ denoted $\rho \models \phi$ iff $\rho,1,\bar{0}\models \phi$. Here $\bar{0}$ 
 is the valuation obtained by setting all freeze variables to 0.  
 We denote by $\mathsf{k{-}\tptl}$ the fragment of  $\mathsf{TPTL}$ using at most $k$ freeze variables.
 The fragment of $\mathsf{TPTL}$ with $k$ freeze variables and using only $\until$ modality is denoted $\mathsf{k{-}\tptl[\until_I]}$.

\section{Satisfiability Checking for Open TPTL}
This section is dedicated to our first contribution, i.e., showing that one variable $\tptl$ does not enjoy the benefits of relaxing punctuality. As mentioned above, one can encode punctual constraints by boolean combination of non-punctual constraints. Hence, to be fair, we further strengthen the notion of non-punctuality to only allow specification over topologically open sets. Hence, the punctual guards are no longer expressible as punctual intervals are closed intervals. To this end we first define the fragment ``Open TPTL'' $\optptl$.

 \subsection{Open $\mathsf{TPTL}$ with 1 Variable ($\optptl$)}
 Open TPTL $\optptl$ is a subclass of $\tptl$ with following restrictions:  
 \begin{itemize}
 	\item For any timing constraint $T-x \in I$ appearing within the scope of even (respectively odd) number of negations, $I$ is an open (respectively closed) interval. 

 \end{itemize}
 Note that this is a stricter restriction than non-punctuality as it can assert a property only within an open timed regions. 
$\optptl[\fut,\nx]$ is a subclass of $\optptl$ where  only $\fut$, $\Box$ and $\nx$ temporal modalities are allowed. 
Similarly, $\optptl[\fut, \past, \nx]$ is a subclass of $\optptl$ where formulae are only allowed $\fut$, $\Box$, $\past$, $\boxminus$ and $\nx$ temporal modalities. The rest of the section is dedicated in proving the following theorem.

 \begin{theorem}
 \label{thm:optptl}
 	Satisfiability Checking of:
 	\begin{enumerate}
 		\item  $\optptl[\fut,\nx]$  over finite timed words is non-primitive recursive hard.
 		\item $\optptl[\fut,\past,\nx]$ is undecidable.
 	\end{enumerate}
 \end{theorem} 

The undecidability of $\optptl[\fut, \past, \nx]$ is via reduction to halting problem of counter machines, and non-primitive recursive hardness for $\optptl[\fut, \nx]$ is via reduction to halting problem of counter machines with incremental errors.
 
\subsection{Counter Machines}
A deterministic $k$-counter machine is a  $k+1$ tuple ${\cal M} = (P,C_1,\ldots,C_k)$, where
\begin{enumerate}
	\item  $C_1,\ldots,C_k$ are counters taking values in  $\mathbb{N} \cup \{0\}$ (their initial values  are set to zero);
	\item  $P$ is a finite set of instructions with labels $p_1, \dots, p_{n-1},p_n$. 
	There is a unique instruction labeled HALT. For $E \in \{C_1,\ldots,C_k\}$, the instructions in $P$ are of the following forms:
	\begin{enumerate} 
		\item  $p_i$: $Inc(E)$, goto $p_j$, 
		\item  $p_i$: If $E =0$, goto $p_j$, else go to $p_k$, 
		\item $p_i$: $Dec(E)$, goto $p_j$,
		\item  $p_n$: HALT. 
	\end{enumerate}
\end{enumerate}
A configuration $W=(i,c_1,\ldots,c_k)$ of ${\cal M}$ is given by the value of the current program counter $i$ and values $c_1,c_2,\ldots,c_k$ of the counters $C_1,C_2,\ldots,C_k$.  A move of the  counter machine $(l,c_1,c_2,\ldots,c_k) \rightarrow (l',c_1',c_2',\ldots,c_k')$ denotes that configuration $(l',c_1',c_2',\ldots,c_k')$ is obtained from $(l,c_1,c_2,\ldots,c_k)$ by executing the $l^{th} $ instruction $p_l$.
If $p_l$ is an increment (or a decrement) instruction of the form ``$INC(c_x)$($Dec(c_x)$, respectively), goto $p_j$'', then $c'_x=c_x+1$ (or $c_x-1$, respectively), while $c'_i=c_i$ for $i \neq x$ and $p'_l = p_j$ is the respective 
next instruction. 
While if $p_l$ is a zero check instruction of the form ``If $c_x =0$, goto $p_j$, else go to $p_k$'', then $c'_x=c_x$ for all $i$, and $p'_l=p_j$ if $c_x=0$ else $p'_l = p_k$.
\begin{theorem}
	\label{theo:minsky}
	\cite{minsky} The halting problem for deterministic $k$-counter machines is undecidable for $k \ge 2$.
\end{theorem}

\subsection{Incremental Error Counter Machine ($\iecm$)}
An incremental error counter machine ($\iecm$) is a counter machine where a particular configuration can have counter values with arbitrary positive error.
Formally, an  incremental error $k$-counter machine is a $k+1$ tuple ${\cal M} = (P,C_1,\ldots,C_k)$ where $P$ is a set of instructions like above and $C_1$ to $C_k$ are the counters.  The difference between a counter machine with and without incremental counter error is as follows:
\begin{enumerate}
\item Let $(l,c_1,c_2\ldots,c_k) \rightarrow (l',c_1',c_2'\ldots,c_k')$ be a move of a counter machine without error when executing $l^{th}$ instruction.
\item The corresponding move in the increment error counter machine is 
$$(l,c_1,c_2\ldots,c_k) \rightarrow \{(l',c_1'',c_2''\ldots,c_k'') | c_i'' \ge c_i', 1 \leq i \leq k \}$$
 Thus the value of the counters are non deterministic. We use these machines for proving lower bound complexity in section \ref{open-tptl}.
\end{enumerate}
\begin{theorem}
	\label{theo:lazic}
	\cite{demriL06} The halting problem for incremental error $k$-counter machines is non primitive recursive.
\end{theorem}

\subsection{Proof of Theorem \ref{thm:optptl}}
\label{open-tptl}

\begin{proof}
\begin{enumerate}
\item We encode the runs of any given incremental error $k$-counter machine using $\optptl$ formulae with $\fut, \nx$ modalities. 
We will encode a particular computation of any counter machine using timed words. The main idea is to construct an $\optptl[\fut,\nx]$ formula $\varphi_{\iecm}$ for any given incremental error $k$-counter machine $\iecm$ such that $\varphi_{\iecm}$ is satisfied by only those timed words that encode the halting computation of $\iecm$. 
Moreover, for every halting computation $\mathcal{C}$ of the $\iecm$, at least one timed word $\rho_C$ encodes $\mathcal{C}$
and satisfies $\varphi_{\iecm}$. 

We encode each computation of a incremental error $k$-counter machine, $(P,C)$ where $P = \{p_1,\ldots,p_n\}$ is the set of instructions and $C = \{c_1,\ldots,c_k\}$ 
is the set of counters   using 
timed words over the alphabet $\Sigma_{\iecm}= \bigcup_{j \in \{1,\ldots,k\}} (S \cup F \cup \{a_j,b_j\})$ where $S = \{s^{p}|p \in {1,\ldots,n}\}$ and $F = \{f^{p}|p \in {1,\ldots,n}\}$ as follows:
\\The  $i^{th}$ configuration, $(p,c_1,\ldots,c_k)$ is encoded  in the timed region $[i,i+1)$ with the sequence 
\begin{quote}
	$s^{p} (a_1b_1)^{c_1} (a_2b_2)^{c_2} \ldots(a_kb_k)^{c_k} f^{p}$.
\end{quote}

\begin{figure}[h]
\includegraphics[scale=0.5]{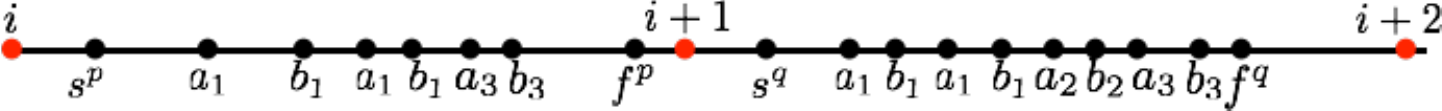}	
\label{conf1}
\caption{Assume there are 3 counters, and that the $i$th configuration is $(p, 2,0,1)$.  Let the instruction $p$ increment counter 2 and go to 
instruction $q$. Then the $i+1^{th}$ configuration is $(q,2,1,1)$. Note that the $i^{th}$ configuration 
is encoded between integer points $i, i+1$, while configuration $i+1$ is encoded between 
integer points $i+1, i+2$.}
\end{figure}

The concatenation of these time segments of a timed word encodes the whole computation.
Untimed behaviour of the language described above is in accordance with the following pattern:
	$$ (\mathcal{S} (a_1b_1)^* (a_2b_2)^* \ldots (a_kb_k)^*\mathcal{F})^*$$

where $\mathcal{S} = s^p$ iff $\mathcal{F} = f^p$ for $p \in \{1,2,\ldots,n\}$.

To construct a formula $\varphi_{\iecm}$, the main challenge is to propagate the behaviour from the time segment 
$[i,i+1)$ to the time segment $[i+1,i+2)$ such that the latter encodes the $i+1^{th}$ configuration of the input $\iecm$ given that the former encodes the $i^{th}$ configuration. 
The usual idea is to copy all the $a$'s from one configuration to another using punctuality. This is not possible in a non-punctual logic.
We preserve the number of $a$s and $b$s using the following idea:
\begin{itemize}
	\item Given any non last $(a_j,t)(b_j,t')$ before $\mathcal{F}$(for some counter $c_j$), of a timed word encoding a computation.
	We assert that the last symbol in $(t,t+1)$ is $a_j$ and the last symbol in 
	$(t',t'+1$) is $b_j$. 
	\item We can easily assert that the untimed sequence of the timed word is of the form
		$$ (\mathcal{S} (a_1b_1)^* (a_2b_2)^* \ldots (a_kb_k)^*\mathcal{F})^*$$
	\item The above two conditions imply that there is at least one $b_j$ within time $(t+1,t'+1)$. Thus, all the non last $a_j,b_j$ are copied to the segment encoding next configuration. 
	Now appending one $a_jb_j$, two $a_jb_j$'s or no $a_jb_j$'s  depends on whether the instruction was copy, increment or decrement operation. 
\end{itemize}

	$\varphi_{\iecm}$ is obtained as a conjunction of several formulae. Let $\mathcal{S},\mathcal{F}$ be a shorthand for 
$\bigwedge \limits_{p\in \{1,\ldots,n\}}s^{p}$ and  $\bigwedge \limits_{p\in \{1,\ldots,n\}}f^{p}$, respectively. We also define macros $A_j = \bigvee \limits_{w \ge j} a_w$ and $A_{k+1} = \bot$
		 We now give formula for encoding the machine. Let $\Cc=\{1,\ldots,k\}$ and $\Pp=\{1,\ldots,n\}$ be the indices of the counters and the instructions.  
	\begin{itemize}
		\item \textbf{Expressing untimed sequence}: The words should be of the form
		 $$ (\mathcal{S} (a_1b_1)^* (a_2b_2)^* \ldots (a_kb_k)^*\mathcal{F})^*$$
		This could be expressed in the formula below
		\begin{quote}

			$\varphi_1 = \bigwedge \limits_{j \in \Cc, p \in \Pp} 
			\wB[s^p \rightarrow \nex(A_1 \vee f^p)] \wedge \wB[a_j \rightarrow \nex(b_j)] \wedge 
			\\
			~~~~~~~~~~~~~~~\wB[b_j \rightarrow \nex(A_{j+1} \vee f^p)] 
			\wedge 
			\wB[f^p \rightarrow \nex (\mathcal{S} \vee \wB(\false))]$
		\end{quote}

	\item \textbf{Initial Configuration}: There is no occurrence of $a_jb_j$ within $[0,1]$. The program counter value is $1$.
		\begin{quote}
			$\varphi_2 = x.\{s^{1} \wedge \nex(f^{1} \wedge x \in (0,1))\} $
		\end{quote}
		\item \textbf{Copying $\mathcal{S},\mathcal{F}$}: Any $(\mathcal{S},u)$ (read as any symbol from 
		$\mathcal{S}$ at time stamp $u$) 
		 $(\mathcal{F},v)$ (read as any symbol from 
		$\mathcal{F}$ at time stamp $v$) has a next occurrence $(\mathcal{S},u')$, $(\mathcal{F},v')$ in the future such that 
		$u' - u \in (k,k+1)$ and $v' - v \in (k-1,k)$. 
		Note that this condition along with $\varphi_1$ and $\varphi_2$ makes sure that $\mathcal{S}$ and $\mathcal{F}$ occur only 
		within the intervals of the form $[i,i+1)$ where $i$ is the configuration number. Recall that $s^n,f^n$
		represents the last instruction (HALT). 
		\begin{quote}
			$\varphi_3 = \wB x.\{(\mathcal{S} \wedge \neg s^{n}) \rightarrow \neg \fut(x \in [0,1] \wedge \mathcal{S}) \wedge \fut(\mathcal{S} \wedge x \in (1,2))\} \wedge 
			 \wB x.\{(\mathcal{F} \wedge \neg f^{n}) \rightarrow \fut(\mathcal{F} \wedge x \in (0,1))\}$
		\end{quote}
	Note that the above formula ensures that subsequent configurations are encoded in 
	smaller and smaller regions within their respective unit intervals, since consecutive 
	symbols from $\mathcal{S}$ grow apart from each other (a distance $>1$), while consecutive
	symbols from $\mathcal{F}$  grow closer to each other (a distance $<1$). See Figure \ref{conf2}.
		 \begin{figure}[h]
\includegraphics[scale=0.5]{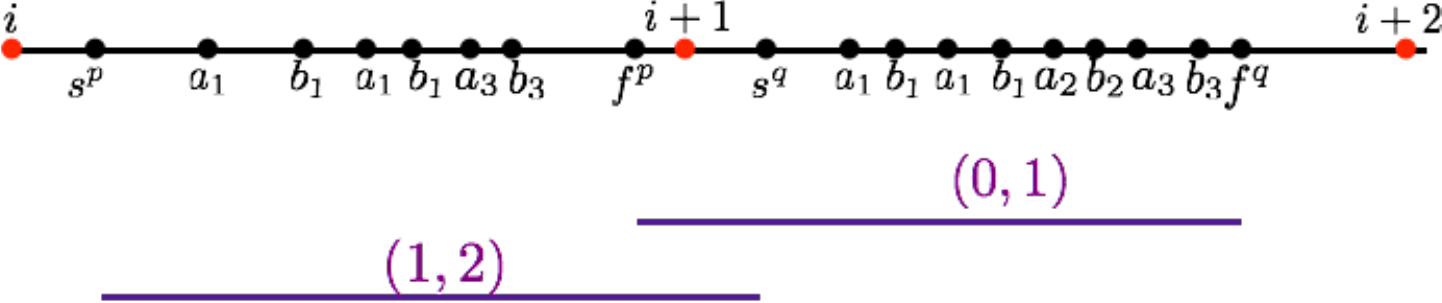}	
\caption{Subsequent configurations in subsequent unit intervals grow closer and closer.}
\label{conf2}
\end{figure}

		\item  Beyond $p_n$=HALT, there are no instructions
		\begin{quote}
		$\varphi_{4}\ =\  \wB[f^{n} \rightarrow \Box(\false)]$
	\end{quote}
		
		\item At any point of time, exactly one event takes place. Events have distinct time stamps.
		\begin{quote}
		$\varphi_6\ =\ [\bigwedge \limits_{y \in \Sigma_{\iecm}}\wB[y \rightarrow
		\neg(\bigwedge \limits_{ \Sigma_{\iecm} \setminus \{y\}}(x))] \wedge \wB[\Box(\false) \vee \nex(x \in (0,\infty))]
		$
		\end{quote}
		\item Eventually we reach the halting configuration $\langle p_n,c_1,\ldots,c_k \rangle$: $\varphi_6 = {\wF} s^{n}$\\
		
		\item Every non last $(a_j,t)(b_j,t')$ occurring in the interval $(i,i+1)$  should be copied in the interval $(i+1,i+2)$. We specify this 
		condition as follows:
		\begin{itemize}
		\item state that from every non last $a_j$   the last symbol within $(0,1)$ is $a_j$.
		Similarly from every non last $b_j$, the last symbol within $(0,1)$ is $b_j$.
		Thus  $(a_j,t)(b_j,t')$ will have a $(b_j,t'+1-\epsilon)$ where $\epsilon \in(0,t'-t)$. 
		
		 \begin{figure}[h]
\includegraphics[scale=0.5]{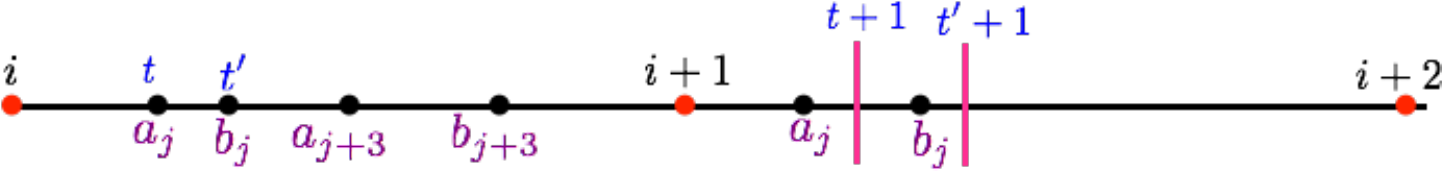}	
\caption{Consider any $a_jb_j$, where $a_j$ is at time $t$ and $b_j$ is at time $t'$. There are further $a,b$ symbols 
in the unit interval, like as shown above $a_{j+3}b_{j+3}$ occur after $a_jb_j$ in the same unit interval. 
Then the $a_j, b_j$ are copied such that the last symbol in the interval $(t, t+1)$ is an $a_j$ 
and the last symbol in the interval $(t', t'+1)$ is a $b_j$. There are no points 
between the  $a_j$ shown above in $(i+1, i+2)$ and the time stamp $t+1$ as shown above. Likewise, 
there are no points between the $b_j$ shown above in $(i+1, i+2)$ and the time stamp $t'+1$ as shown above. 
Note that the time stamp of the copied $b_j$ in $(i+1, i+2)$ lies in the interval $(t+1, t'+1)$.}
\label{conf3}
\end{figure}
\item Thus all the non last $a_jb_j$ will incur a $b_j$ in the next configuration. $\varphi_1$ makes sure that there is an $a_j$ between two $b_j$'s. Thus, this condition along with $\varphi_1$, makes sure that the non last $a_jb_j$ sequence is conserved. Note that there can be some $a_jb_j$ which are arbitrarily inserted. These insertions errors model the incremental error of the machine.
Any such inserted $(a_j,t_{ins})(b_j,t_{ins}')$ in $(i+1, i+2)$ is such that there is a $(a_j,t)(b_j,t')$ in $(i, i+1)$ with $t_{ins}' \in (t+1,t'+1)$. 
Just for the sake of simplicity we assume that $a_{k+1}= false$.
\end{itemize}
Let $nl(a_j) = a_j \wedge \neg last(a_j)$, $nl(b_j) = b_j \wedge \neg last(b_j)$, $\psi_{nh} = \neg \fut (f^n \wedge x\in[0,1])$, \\
$last(a_j) = a_j \wedge \nex(\nex (\mathcal{F} \vee A_{j+1})))$ and $last(b_j) =  b_{j} \wedge \nex(\mathcal{F} \vee A_{j+1})$. 
\begin{quote}
$\varphi_7 = \bigwedge \limits_{j \in \Cc} \wB x.[(nl(a_j) \wedge \psi_{nh}) \rightarrow \fut(a_j \wedge x \in 
(0,1) \wedge \nex(x \in  (1,2)))]  \wedge$\\$ 
\wB x.[(nl(b_j) \wedge \psi_{nh}) \rightarrow \fut(b_j \wedge x \in  (0,1) \wedge \nex(x \in  (1,2)))]$
\end{quote}
\end{itemize}
We define a short macro 		
$Copy_{\Cc \setminus W}$: Copies the content of all the intervals encoding counter values except counters in $W$. Just for the sake of simplicity we define the following notation 
\begin{quote}
$Copy_{\Cc  \setminus W} = \bigwedge \limits_{j \in \Cc \setminus W} \wB x.\{last(a_j) \rightarrow (a_j \wedge x \in (0,1) \wedge \nex(b_j \wedge x \in (1,2) \wedge \nex (\mathcal{F})))\}$
\end{quote}	
Using this macro we define the increment,decrement and jump operation.

\begin{enumerate}
	\item Consider the zero check instruction $p_g$: If $C_j=0$ goto $p_h$, else goto $p_d$. $\delta _1$ specifies the next configuration when the check for zero succeeds. $\delta_2$ specifies the else condition.
	\begin{quote}
		$ 
		\varphi^{g,j=0}_{8}\ =\  Copy_{\Cc \setminus \{\emptyset\}}
	\wedge \delta_1 \wedge \delta_2$
	\end{quote}
	\begin{quote}
		$\delta_1=	\wB[\{s^g \wedge ((\neg a_j) \until \mathcal{F})\} \rightarrow (\neg \mathcal{S}) \until s^h)]$
		
		$\delta_2 =   \wB[\{s^g \wedge ((\neg a_j) \until a_j )\} \rightarrow (\neg \mathcal{S}) \until s^d)]$.
		
	\end{quote}
	\item $p_g$: $Inc(C_j)$ goto $p_h$. The increment is modelled by appending exactly one $a_jb_j$ in the next interval just after the last copied $a_jb_j$
	\begin{quote}
		$ 
		\varphi^{g,inc_j}_{8}\ = Copy_{\Cc \setminus \emptyset} \wedge \wB(s^g \rightarrow (\neg \mathcal{S}) \until s^h )\wedge \psi^{inc}_{0} \wedge \psi^{inc}_{1}$
		\end{quote}
		\begin{itemize}
		\item The formula $\psi^{inc}_{0} = \wB [(s^g \wedge (\neg a_j \until f^g)) \rightarrow (\neg \mathcal{S} \until x.(s^h \wedge \fut(x \in (0,1) \wedge a_j))]$ specifies the increment of the counter $j$ when the value of $j$ is zero. 
		\item The formula
		$\psi^{inc}_{1}=\wB [\{s^g \wedge ((\neg \mathcal{F}) \until (a_j))\} \rightarrow  (\neg \mathcal{F}) \until x.\{last(a_j) \wedge \fut(x \in (0,1) \wedge (a_j \wedge \nex\nex (last(a_j) \wedge x \in (1,2))))\} ]$
	specifies the increment of counter $j$ when $j$ value is non zero by appending exactly one pair of $a_jb_j$ after the last copied $a_jb_j$ in the next interval.
		 	
		\end{itemize}

		\item $p_g$: $Dec(C_j)$ goto $p_h$. Let $\mathsf{second-last}(a_j)= a_j \wedge \nex(\nex(last(a_j)))$. Decrement is modelled by avoiding copy of last $a_jb_j$ in the next interval. 
			\begin{quote}
			$ 
				\varphi^{g,dec_j}_{8}\ =\   Copy_{\Cc \setminus j} \wedge \wB(s^g \rightarrow (\neg \mathcal{S}) \until s^h )\wedge \psi^{dec}_0
		        \wedge \psi^{dec}_1
			$
			\end {quote}
			\begin{itemize}
			\item 	
			The formula $\psi^{dec}_0 = \wB [\{s^g \wedge (\neg a_j) \until f^g)\} \rightarrow \{(\neg \mathcal{S}) \until \{s^h \wedge ((\neg a_j) \until (\mathcal{F})\}]$ specifies that the counter remains unchanged if decrement is applied to the $j$ when it is zero. 
		\item 	The formula $\psi^{dec}_1 = \wB [\{s^g \wedge ((\neg \mathcal{F}) \until (a_j))\} \rightarrow  (\neg \mathcal{F}) \until x.\{\mathsf{second-last}(a_j) \wedge \fut(x \in (0,1) \wedge (a_j \wedge \nex\nex ([A_{j+1} \vee \mathcal{F}] \wedge x \in (1,2))))\} ]$ decrements the counter $j$, if the present value of $j$ is non zero. It does that by disallowing copy of last $a_jb_j$ of the present interval to the next. 
			
			\end{itemize}

\end{enumerate}

The formula $\varphi_{\iecm}= \bigwedge \limits_{i \in \{1,\ldots,7\}} \varphi_i \wedge \bigwedge \limits_{p \in \Pp} \varphi^{p}_8$.

 \item To prove the undecidability we encode the k counter machine without error. Let the formula be $\varphi_{\mathsf{CM}}$. The encoding is same as above. The only difference is while copying the non-last $a$ in the $\varphi_{\cal M}$ we allowed insertion errors i.e. there were arbitrarily extra $a$ and $b$ allowed in between apart from the copied ones in the next configuration while copying the non-last $a$ and $b$.
 To encode counter machine without error we need to take care of insertion errors. 
 Rest of the formula are same. The following formula will avoid error and copy all the non-last $a$ and $b$ without any extra $a$ and $b$ inserted in between.
\begin{quote}
	$\varphi_{9} = \bigwedge \limits_{j \in C} \wB x.[(a_j \wedge \neg last(a_j)) \rightarrow \past(x-T \in (1,2) \wedge \nex(a_j \wedge x-T \in (0,1)))] \wedge 
	\\ ~~~~~~~~~~~~~~~~~~\wB x.[(b_k \wedge \neg last(b_j)) \rightarrow \past(x-T \in (1,2) \wedge \nex(b_j \wedge x-T \in (0,1)))]$
\end{quote}
Now,
$\varphi_{\mathsf{CM}} = \varphi_{\iecm} \wedge \varphi_{9}$

\end{enumerate}

\paragraph*{Correctness Argument}
 Note that incremental error occurs only while copying the non last $a_jb_j$ sequence. The similar argument for mapping $a_j$ with a unique $a_j$ in the next configuration can be applied in past. Note that $\varphi_7$ makes sure that for every non-last $(a_j,\tau_x) (b_j,\tau_{x+1})$ there is at least one $b_j$ with timestamp $\tau_y \in (\tau_x,\tau_{x+1})+1$. Thus, there exists an injective mapping from the $b_j$'s copied in the interval $[i+1,i+2)$ to consecutive pair of non-last $a_j,b_j$ in the interval $[i,i+1)$. 
 Note that the formulae $\varphi_9$ symmetrically maps every non-last $b_j$ in the interval $[i,i+1)$ to a pair of copied $a_j,b_j$ in the interval $[i+1, i+2)$. As there is an injection in both the directions, there is a bijection between all the non-last $a_j,b_j$ sequence in $[i,i+1)$ and set of copied $a_j,b_j$ sequence in $[i+1,i+2)$. Thus, $\varphi_9$ implies that there are no insertion errors.
\end{proof}

\section{Satisfiability Checking for Partially Adjacent 1-TPTL}

We start with the definition of non-adjacency proposed by the authors in \cite{KKMP21} \cite{KKMP23}. 
Given any interval $I \in \intinterval$, let $\sup (I)$ (respectively $\inf(I)$) be the supremum  (respectively infimum) of interval $I$. Let $I, I' \in \intinterval$ be any pair of intervals. We say that $I$ is adjacent to $I'$ iff $\sup(I) = \inf(I')$ implies $\sup(I) = 0$ and $\inf(I) = \sup(I')$ implies $\inf(I) = 0$. Intuitively, the \textbf{non-zero} supremum (respectively infimum) of one interval should not be the same as the non-zero infimum (respectively supremum) of the other interval.

Before defining the notion of Partial Adjacency, we introduce the notion of positive non-adjacency and negative non-adjacency. $I$ and $I'$ are said to be \textbf{positively non-adjacent} to each other iff $\sup(I) = \inf(I')$ implies $\sup(I) \le 0$ and $\inf(I) = \sup(I')$ implies $\inf(I) \le 0$. In other words, non-negative supremum (respectively infimum) of one interval should not be the same as the infimum (respectively supremum) of the other interval.
Similarly,  $I$ and $I'$ are said to be \textbf{negatively non-adjacent} to each other iff $\sup(I) = \inf(I')$ implies $\sup(I) \ge 0$ and $\inf(I) = \sup(I')$ implies $\inf(I) \ge 0$.

We say that a set of intervals $\mathcal{I} \subseteq \intinterval$ is non-adjacent, positively non-adjacent, and negatively non-adjacent iff for all $I,I' \in \mathcal{I}$ $I$ and $I'$ are non-adjacent, positively non-adjacent, and negative non-adjacent, respectively.

Notice that any punctual interval of the form $\langle x, x \rangle$ where $x \ne 0$ is not adjacent to itself. Hence, if a set of intervals $\mathcal{I}$ is non-adjacent, it will not contain any such punctual intervals. We now define Partially Adjacent TPTL.

\subsection{Partially Adjacent TPTL}
\label{sec:patptl}
Before defining Partially Adjacent TPTL, we first define the negation normal form for TPTL. In   this form negations are pushed inside and are only applicable to the propositional subformula.
$\tptl$ in the negation normal form is defined by the following grammar:
$$\varphi::=a~|~\neg a ~|\top~|~\bot~|~ x.\varphi ~|~ T-x \in I ~|~x-T \in I~|~\varphi \wedge \varphi~|~ \varphi \vee \varphi~|
~\varphi \until \varphi ~|~ \varphi \since \varphi ~|~ \Box \varphi ~|~ \boxminus \varphi,$$
where $x \in X$, $a \in \Sigma$, $I \in \intinterval$.
It is routine to show that any TPTL formula can be reduced to an equivalent formula in negation normal form with linear blow-up. Henceforth, without loss of generality, we will assume that a given TPTL formula is in negation normal form.

We now define the Non-Adjacent TPTL of \cite{KKMP21}, followed by the definition of Partially Adjacent TPTL.

\subsubsection{Non-Adjacent TPTL}
Given any $1$-TPTL formula $\varphi$ in negation normal form, we say that $\varphi$ is a non-adjacent TPTL $\natptl$ formula iff any pair of intervals used in the constraints appearing within the scope of the same freeze quantifier are non-adjacent. For example $x. (a \until (b \wedge T-x \in (1,2) \wedge T-x \in (3,4)))$ is a non-adjacent TPTL formula while, 
 $x. (a \until (b \wedge T-x \in (1,2) \wedge T-x \in (2,3)))$ and  $x. (a \until (b \wedge T-x \in [2,2]))$ are not non-adjacent TPTL formula.
 Similarly, $x. \{a \until \{b \wedge T-x \in [1,2] \wedge x. (c \until (d \wedge T-x \in [2,3])) \}\}$ is a non-adjacent TPTL formula in spite of intervals $[1,2]$ and $[2,3]$ being non-adjacent, because the non-adjacent intervals in these formula do not appear within the scope of the same freeze quantifier.

 \subsubsection{Partially-Adjacent TPTL}
 Similarly, given any $1$-TPTL formula $\varphi$ in negation normal form, we say that $\varphi$ is a positively (respectively negatively) non-adjacent TPTL $\pnatptl$ (respectively $\nnatptl$) formula iff any pair of intervals used in the constraints appearing within the scope of the same freeze quantifier are positively non-adjacent (respectively negatively non-adjacent). For example,  $x. \{a \until \{b \wedge T-x \in [1,2] \wedge (c \until (d \wedge T-x \in [2,3])) \}\}$ is an  $\nnatptl$ formula but not a $\pnatptl$ formula.  $x. \{a \until \{b \wedge T-x \in [1,2] \wedge T-x \in [2,3] \wedge  (c \since (d \wedge T- x \in [-1,-3])) \}\}$ is a $\nnatptl$ formula. Similarly, $x. \{a \until \{b \wedge T-x \in [1,2] \wedge (c \since (d \wedge T- x \in [-1,-3] \wedge T-x \in [0,-1])) \}\}$ is a $\pnatptl$ formula. Finally, $x. \{a \until \{b \wedge T-x \in [1,2] \wedge T-x \in [2,3] \wedge (c \since (d \wedge T- x \in [-1,-3] \wedge T-x \in [0,-1])) \}\}$ is neither $\pnatptl$ nor $\nnatptl$ formula. Notice that a formula $\varphi$ is both $\pnatptl$ and $\nnatptl$ iff $\varphi$ is a $\natptl$ formula. Hence, $\natptl$ is a syntactic subclass of both $\pnatptl$ and $\nnatptl$. It is also a strict subclass of both these logics as the latter can also express certain punctual constraints.
 It should also be straightforward that $\mtlpwuisnp$ is a strict syntactic subclass of $\pnatptl$ and $\mtlpwunpsi$ is a strict syntactic subclass of $\nnatptl$. 

 Any formula is said to be a \textbf{partially-adjacent TPTL} $\patptl$ iff it is either a $\pnatptl$ formula or $\nnatptl$ formula.

\subsection{Satisfiability Checking for Partially-Adjacent TPTL}
\label{sec:patptlsat}
This section is dedicated to the proof of our second contribution, i.e., the following theorem.
\begin{theorem}
\label{thm:patptl}
    Satisfiability Checking for Partially-Adjacent TPTL formulae is decidable over finite timed words.
\end{theorem}
We show the decidability for $\nnatptl$ formula. The decidability proof for $\pnatptl$ is symmetric. The proof of the above theorem is via satisfiability preserving reduction to $\mtlpwuisnp$. 

Before we start with the proof, we first define the intermediate logics which are used in this reduction.

\subsubsection{MTL with Regular Expression Guarded Modalities, RatMTL}
We now define MTL extended with Regular Expression Guarded modalities from \cite{KKP17}. The modalities of these logics are in extensions of $\emitl$ of \cite{Wilke}. The latter used automata over subformulae rather than regular expressions.

Given a finite alphabet $\Sigma$,  the formulae of $\regmtl$ are built from following grammar:
$$\varphi ::\begin{array}{r}
     a (\in \Sigma)~|true~|\varphi \wedge \varphi~|~\neg \varphi~|
~\regm_I (\re \langle \varphi_1,\ldots,\varphi_n\rangle ) ~|~ \fregm_I (\re) (S) ~|~ \pregm_I (\re) (S)
\end{array}$$
\\ where $S$ is a finite set of subformulae generated using the above grammar, $I \in \intintervaln$ and $\re$ is a regular expression over subformulae in $S$. 
\\{\bf Semantics}: For any timed word $\rho = (a_1, \tau_1),(a_2, \tau_2) \ldots (a_n, \tau_n)$ let $x,y \in dom(\rho)$ $x<y$ and let $\Ss = \{\varphi_1,\ldots, \varphi_n\}$ be a given set of formulae ($\emitl, \regmtl,\mitl$ etc.). We define 
$\mathsf{Seg^+}(\rho,x,y,\Ss)$ (respectively $\mathsf{Seg^-}(\rho,x,y,\Ss)$)  as untimed words $(b_{x} b_{x+1} \ldots  b_{y})$ in $\Ss^*$ (respectively $(b_{y} b_{y-1} \ldots  b_{x})$ in $\Ss^*$) such that for any $x\le z \le y$, $b_z = S'$ iff $\rho,z$ satisfies all the formulae in $S'$ and none of the formulae in $\Ss \setminus S'$ (i.e. $\rho,z \models \bigwedge S' \wedge \neg \bigvee (S' \setminus \Ss)$). 

Similarly we define $\mathsf{TSeg}(\rho,i,I,\Ss)$ as $\mathsf{Seg^+}(\rho,x,y,\Ss)$ where $x > i$ and  
$\tau_x ,\tau_y \in I+\tau_i$, and either $y=|dom(\rho)|$ or  $\tau_{y+1} \notin I+\tau_i$, and when $I$ is of the form $[0,u\rangle$ then  $x = i+1$ else $\tau_{x-1}\notin I$. 
In other words, $\mathsf{TSeg}(\rho,i,I,\Ss)$ is an untimed word storing the truth values of   the formulae in $\Ss$ for the sub-word of $\rho$ within interval $I$ from point $i$. 

For example, consider a timed word $\rho = (a,0) (b,0.5) (b,0.95) (b, 1.9)$ and consider set of $\mitl$ formulae $\Ss = \{\fut_{(0,1)} b, \fut_{(1,2)} b\}$. In this case $\mathsf{Seg^+}(\rho, 1, 3, \Ss) = P[\Ss]P[\Ss] P[{\{\fut_{(0,1)} b\}}]$ as the first and the second point satisfies both the formulae in $\Ss$ but the third point satisfies only $\fut_{(0,1)} b$ as there is no point within time interval $(1,2)$ in its future. Similarly, $\mathsf{TSeg}(\rho,1,(0,1),\Ss) = \mathsf{Seg^+}(\rho, 2, 3, \Ss) = P[\Ss] P[\{\fut_{(0,1)} b\}]$.
 
Semantics for propositional formulae and boolean combinations are defined as usual. We just define the semantics of regular expression modalities:\\
1) $\rho,i \models \regm_I (\re)(S) $  $\leftrightarrow$  $[\mathsf{TSeg}(\rho, i, I, S)] \in L(\re)$\\
2) $\rho,i \models \fregm_I (\re)(S)$  $\leftrightarrow$  $\exists j \ge i, \tau_j - \tau_i \in I \wedge $ $[\mathsf{Seg^+}(\rho, i+1, j, S)] \in L(\mathcal{\re})$\\
3) $\rho,i \models \pregm_I (\re) (S)$  $\leftrightarrow$  $\exists j \le i, \tau_j - \tau_i \in I \wedge $ $[\mathsf{Seg^-}(\rho, i+1, j, S)] \in L(\mathcal{\re})$\\

Following are examples of some properties in $\regmtl$.
\begin{itemize}
\item {Example 1}: $\fregm_{(1,2)}((\{a\}.\{a\})^*.\{b\}) (\{a, b\})$, 
states that $a$ is true until the first $b$ in the future from the present point is within time interval $(1,2)$ and the number of points where $a$ is true before the first $b$ in the future is even.\\
\item Example 2: $\until_I$ can be easily defined using $\fregm_I$ modality. $$\varphi_1 \until_I \varphi_2 \equiv \freg_I ([\{\varphi_1\}+ \{\varphi_1, \varphi_2\})^*.(\{\varphi_2\}+ \{\varphi_1, \varphi_2\})]) (\{\varphi_1, \varphi_2\}).$$As $\until$ can be efficiently written in $\freg_I$, we don't include it in the syntax of $\regmtl$.
\item Example 3: $\regm_{[1,2]}((\{a\}.\{b\})^*) (a, b)$, specifies that the behaviour within interval $[1,2]$ from the present point, is such that $a$ holds at the first point within $[1,2]$, $b$ holds at the last point within $[1,2]$, and $a$ and $b$ alternately hold within $[1,2]$.
\end{itemize}

\begin{theorem} \cite{KKP17}\cite{khushraj-thesis}
    $\regm$ modality is more expressive than $\fregm$ modality.
\end{theorem}

As both NFA and regular expressions are equivalent in terms of expressiveness, we assume that all the arguments within $\fregmk$ and $\regm$ modalities are NFA rather than regular expressions.

\renewcommand{\re}{\mathsf{A}}

\subsubsection{Metric Temporal Logic Extend with Pnueli Automata Modalities \cite{KKMP21}\cite{KKMP23}}
We now define a recent logic introduced by \cite{KKMP21}\cite{KKMP23}. The logic is an extension of MTL with Automata modalities. It syntactically generalizes the automata modalities in EMITL of \cite{Wilke} and $\fregm_I$ modality of $\regmtl$ of \cite{KKP17}. Given a finite alphabet $\Sigma$, formulae of $\pnregmtl$ have the following syntax: \\
$\varphi::{=}a~|\varphi \wedge \varphi~|~\neg \varphi~| \fregmk^k_{I_1,\ldots,I_k} (\re_1,\ldots, \re_{k+1})(S)~|~ 
\pregmk^k_{I_1,\ldots,I_k} (\re_1,\ldots,\re_{k+1})(S)$,\\ where $a \in \Sigma$, 
$I_1, I_2, \ldots I_k \in \intintervaln$ and $\re_1, \ldots \re_{k+1}$ are automata over $2^S$ where $S$ is a set of formulae from $\pnregmtl$. $\fregm^k$ and $\pregm^k$ are the new modalities called future and past {\bf Pnueli Automata} Modalities, respectively, where $k$ is the arity of these modalities.

The satisfaction relation for $\rho,i_0$ satisfying a $\pnregmtl$ formula $\varphi$ is defined recursively as follows:
\begin{itemize}
\item $\rho,i_0{\models}\fregmk^k_{I_1, \ldots, I_k}(\re_1,\ldots,\re_{k+1})(S)$ iff 
 ${\exists} {i_0{ {\le} }i_1{\le} i_2 \ldots {\le} i_k {\le} n}$ s.t.\\ $\bigwedge \limits_{w{=}1}^{k}{[(\tau_{i_w} {-} \tau_{i_0} \in I_w)}
 \wedge \mathsf{Seg^+}(\rho, i_{w{-}1}+1, i_w, S) \in L({\re_w})] \wedge 
 \mathsf{Seg^+}(\rho, i_{k}+1, n,S) \in L({\re_{k+1}}) $,
\item $\rho,i_0 \models \pregmk^k_{I_1,I_2,\ldots,I_k} (\re_1,\ldots,\re_k, \re_{k+1})(S)$ iff  
${\exists} i_0 {\ge} i_1 {\ge} i_2 \ldots {\ge} i_k {\ge} 1$ s.t. \\
$\bigwedge \limits_{w{=}1}^{k}[(\tau_{i_0} {-} \tau_{i_w} \in I_w)
 \wedge \mathsf{Seg^{-}}(\rho, i_{w{-}1}-1, i_w,S) \in L({\re_{w}})] \wedge 
 \mathsf{Seg^{-}}(\rho, i_{k}-1, 1, S) \in L({\re_{k+1}})$.
\end{itemize}
Notice that we have a strict semantics for the above modalities as opposed to the weak ones of \cite{KKMP23}, for the sake of technical convenience. Both the semantics are equivalent. 
Refer to the figure \ref{fig:fregk} for semantics of $\fregmk^k$. Hence, $\fregmk^k$ and $\pregmk^k$ can be seen as a $k$-ary generalization of $\fregm$ and $\pregm$ modalities, respectively.

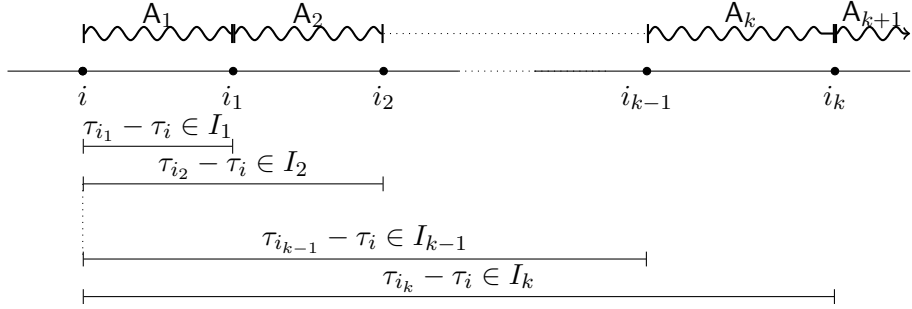
\begin{figure}

	\begin{tikzpicture}
\draw (-1,0)--(5,0);
\draw[dotted](5,0)--(7,0); 
\draw (6,0)--(11,0);

\node[fill = black,draw = black,circle,inner sep=1pt,label=below:{}] at (0,0){};


\node[fill = black,draw = black,circle,inner sep=1pt,label=below:{ $i$}] at (0,0){};
\node[fill = black,draw = black,circle,inner sep=1pt,label=below:{$i_1$}] at (2,0){};
\draw[|-|]  (0,-1)--(2,-1);
\node at (1,-0.75) {$\tau_{i_1} - \tau_i \in I_1$};
\node[fill = black,draw = black,circle,inner sep=1pt,label=below:{$i_2$}] at (4,0){};
\draw[|-|] (0,-1.5)--(4,-1.5);
\node at (2,-1.25) {$\tau_{i_2} - \tau_i \in I_2$};
\draw[dotted](0,-1.5)--(0,-2.5);

\node[fill = black,draw = black,circle,inner sep=1pt,label=below:{$i_{k-1}$}] at (7.5,0){};
\draw[|-|] (0,-2.5)--(7.5,-2.5);
\node at (3.75,-2.25) {$\tau_{i_{k-1}} - \tau_i \in I_{k-1}$};
\node[fill = black,draw = black,circle,inner sep=1pt,label=below:{$i_k$}] at (10,0){};
\draw[|-|] (0,-3)--(10,-3);
\node at (5,-2.75) {$\tau_{i_{k}} - \tau_i \in I_k$};

\draw[|-|,thick]  (0,0.5)--(2,0.5)[decorate, decoration=snake];
\node at (1,0.75){$\mathsf{A}_1$};
\draw[|-|, thick]  (2,0.5)--(4,0.5)[decorate, decoration=snake];
\node at (3,0.75){$\mathsf{A}_2$};
\draw[dotted]  (4,0.5)--(7.5,0.5);
\draw[|-|, thick]  (7.5,0.5)--(10,0.5)[decorate, decoration=snake];
\node at (8.75,0.75){$\mathsf{A}_k$};
\draw[|->, thick]  (10,0.5)--(11,0.5)[decorate, decoration=snake];
\node at (10.5,0.75){$\mathsf{A}_{k+1}$};
\end{tikzpicture}
\caption{Figure showing semantics of $\fregm^k_{I_1,\ldots,I_k}(\re_1, \re_2, \ldots, \re_{k},\re_{k+1})(S)$.}
\label{fig:fregk}
\end{figure}

\subsubsection{Partially Adjacent PnEMTL}
We say that a given $\pnregmtl$ formula is non-adjacent, (respectively positively non-adjacent, negatively non-adjacent) iff $\pnregmtl$ set of intervals within $I_1, \ldots I_k$ appearing associated with any (respectively $\fregmk$, $\pregmk$) modality is such  that $\{I_1, \ldots, I_k\}$ is non-adjacent.
We denote non-adjacent $\pnregmtl$ by NA-$\pnregmtl$, positively non-adjacent $\pnregmtl$ by $\pnaem$, and positively non-adjacent $\pnregmtl$ by $\nnaem$.

\subsection{Proof of Theorem \ref{thm:patptl}}
\label{proof:patptl}
This section is dedicated to the proof of Theorem \ref{thm:patptl}. The proof goes via a satisfiability preserving reduction from the given $\nnatptl$ formula $\varphi$ to $\mtlpwuisnp$ formula $\psi$.
The proof is as follows:
\\(1) We observe that the proof of Theorem 5.15 (Pg 9:23) of \cite{KKMP23} actually proves a more general version of Theorem 5.15 as follows.
\begin{theorem}{General version of Theorem 5.15 of \cite{KKMP23}}
\label{thm1}
    Given any simple 1-$\tptl$ formula $\varphi$ in negation normal form using the set of intervals $I_\nu$, we can construct an equivalent $\pnregmtl$ formula $\alpha$ such that $|\alpha|{=}O(2^{Poly(|\varphi|)})$ and arity of $\alpha$ is at most $2|I_\nu|^2 + 1$. Moreover, any interval $\langle l,u \rangle$ appears within the modality $\fregmk$ in $\alpha$ only if (1) either $l=0$ or there is an interval $\langle l,u' \rangle  \in I_\nu$ and, (2) either $u = 0$ or there is an interval $ \langle l', u \rangle \in I_\nu$. Similarly, any interval $\langle l,u \rangle$ appears within the modality $\pregmk$ in $\alpha$ only if (1) either $l=0$ or there is an interval $\langle l',-l \rangle  \in I_\nu$ and, (2) either $u = 0$ or there is an interval $ \langle -u, u' \rangle \in I_\nu$. Hence, if $\varphi$ is positively (respectively negatively) non-adjacent $\tptl$ formula then $\alpha$ is a $\pnaem$ (respectively $\nnaem$) formula.
\end{theorem}
\noindent (2) The above generalization then implies a general version of Theorem 5.16 of \cite{KKMP23} as follows.
\begin{theorem}
\label{thm2}
Any (respectively non-adjacent, positively non-adjacent, negatively non-adjacent) 1-$\tptl$ formula $\varphi$ with intervals in $I_\nu$, can be reduced to an equivalent (respectively, non-adjacent, positively non-adjacent, negatively non-adjacent) $\pnregmtl$, $\alpha$, with $|\alpha|= 2^{Poly(|\varphi|)}$ and arity of $\alpha{=}O(|I_\nu|^2)$ such that 
$\rho,i \models \varphi$ iff $\rho,i \models \alpha$.
\end{theorem}
\noindent (3) Hence, by theorem \ref{thm2}, we can reduce the given $\nnatptl$ formula $\varphi$ into an equivalent $\nnaem$ formula.
\\(4) Applying the flattening reduction in section 8.2 of \cite{KKMP23}, we get an equisatisfiable formula 
$$\alpha' = \bigvee_i \Box^{ns}(b_i \leftrightarrow [(\bigvee \Sigma) \wedge \beta_i)]$$
where each $\beta_i$ is either a propositional formula, or a formula of the form $\pregmk^k_{I_1, \ldots I_k}(\re_1, \ldots, \re_k)$, or a formula $\fregmk^k_{I_1, \ldots I_k}(\re_1, \ldots, \re_k)$ where each $\re_1, \ldots, \re_k$ is an automata over fresh variables $b_i$. Moreover every $\pregmk^k_{I_1, \ldots I_k}(\re_1, \ldots, \re_k)$ is such that $\{I_1, \ldots, I_k\}$ is a non-adjacent set of intervals.
\\(5) It suffices to reduce these individual conjunct of the form $\Box^{ns}(b_i \leftrightarrow [(\bigvee \Sigma) \wedge \beta_i)]$ to an equivalent $\mtlpwuisnp$ modulo oversampled projection (See Appendix \ref{app:oversampling} for details).
\\(6) By construction in Section 8.3 of \cite{KKMP23}, we can reduce each conjunct of the form $\Box^{ns}(b_i \leftrightarrow [(\bigvee \Sigma) \wedge \beta_i)]$ where $\beta_i$ is a $\pregmk^k$ formula to an equisatisfiable $\emitl_{0,infty}$ formula modulo oversampling. Notice this is possible because all the $\pregmk$ formula are non-adjacent $\pnregmtl$ formula. By lemma 6.3.3 of \cite{khushraj-thesis}, these $\emitl_{0,infty}$ formulae can be converted to an equisatisfiable $\mitl$ formula modulo oversampling.
\\(7) To handle conjunct of the $\Box^{ns}(b_i \leftrightarrow [(\bigvee \Sigma) \wedge \beta_i)]$ where $\beta_i$ is a $\fregmk^k$ formula, we reduce each $\fregmk$ formula to an equivalent $\regm$ formula. By Lemma 6.3.2 of \cite{khushraj-thesis}, these $\regm$ formulae can be reduced to an equisatisfiable $\mtl[\until_I]$ formula modulo oversampling.
\\(8) Thus, it suffices to show that any $\fregmk^k_{I_1, \ldots I_k} (\re_1, \ldots, \re_{k+1})$ can be reduced to an equivalent $\regm$ formula. We drop the second argument $(S)$ from $\regm$ and $\fregmk$ modalities wherever it is clear from context. We assume without loss of generality that $\sup(I_j) \le \inf(I_{j+1})$.

\newcommand{\Seg}{\mathsf{Seg}}
Notice that $\rho, i_0 \models \fregmk^k_{I_1, \ldots I_k} (\re_1, \ldots \re_{k+1})(S)$ iff there exists $i_0 \le i_1 \le i_2 \ldots \le i_k \le i_{k+1} = |\rho|$ such that (A) $\Seg^+(\rho, i_{j-1}+1,i_j, S) \in L(\re_j)$, for $1 \le j \le k+1$. 
Let each $I_j$ be of the form $[ l_j, u_j ]$ \footnote{For the sake of simplicity we assume all the intervals are close intervals. The proof can be generalized to open and half closed intervals, similarly.}. Let $u_{0} = 0$.

Let $i''_{j}$ be the last point within $[u_{j-1},l_j)$ and $i'_{j+1}$ be the last point within interval $[l_j, u_j]$ from the point $i_{0}$. Let $init_j$ be the initial state of $\re_j$. And $F_j$ be the set of accepting states of $\re_j$. Finally, let $\re_j[q,q']$ define an automaton with the same transition relation as $\re_j$ but with the initial state as $q$ and accepting state as $q'$. Similarly, let $\re_j[q, F_j]$ define an automaton with the same transition relation and set of accepting states as $\re_j$ but with the initial state $q$. Let (B1) $q''_1$ be any state of $\re_1$ which is reachable after reading $\Seg^+(\rho, i_{0}+1, i''_1, S)$ starting from the initial state $init_j$. (B2) Let $q_1$ be any state reachable by $\re_1$ on reading 
$\Seg^+(\rho, i''_{1}+1, i_1, S)$ starting from $q_1''$. 
(B3) Let $q_2'$ be any state reachable by $\re_2$ on reading $\Seg^+(\rho, i_{1}+1, i'_2, S)$ starting from initial state.  (B4) Let $q_2''$ be any state reachable by $\re_2$ from after reading $\Seg^+(\rho, i'_{2}+1, i''_2, S)$ starting from the state $q_2'$. (B5) Let $q_2$ be any state reachable by $\re_2$ reading $\Seg^+(\rho, i''_{2}+1, i_3, S)$ starting from the state $q''_2$, and so on. 

Notice that $q_j$  has to be a final state of $\re_j$ iff (A) holds. 

By the definition of $\regm$ modality, the fact that $q''_1$ is reachable from  after reading the word $\Seg^+(\rho, i_{0}+1, i''_1, S)$ starting from initial state of $\re_1$ can be expressed by $\rho, i_0 \models \varphi'_1(q''_1)$, where $\varphi''_1(init_1, q''_1) = \regm_{[0,l_1)}(\re_1[init_1,q''_1])$. Similarly, the fact that there exists a point $i_1 \ge i_0$ such that  starting from state $q''_1$ in $\re_1$ on reading $\Seg^+(\rho, i''_{1}+1, i_1, S)$ we reach some accepting state in $\re_1$,  and $\re_2$ reaches a state $q'_2$ on reading  $\Seg^+(\rho, i_{1}+1, i_2', S)$ can be expressed as $\rho, i_0 \models \varphi_1$, where $\varphi_1(q''_1, q'_2) = \regm_{[l_1,u_1]}(\re_1[q''_1, F_1].\re_2[init_2, q'_2])$. Similarly, the fact that $\re_2$ reaches a state $q_2''$ starting from $q'_2$ on reading $\Seg^+(\rho, i'_{2}+1, i_2'', S)$ can be encoded by $\rho, i_0 \models \varphi''_2(q'_2, q''_2)$ where $\varphi''_2 = \regm_{(u_1,l_2)}(\re_2[q'_2, q''_2])$, and so on. Notice that there are finitely many possible values for $q''_j$ and $q'_j$. Disjuncting over all possible values for $q'_j$ and $q''_j$, we can express $\rho, i_0 \models \fregmk^k_{I_1, \ldots I_k} (\re_1, \ldots \re_{k+1})(S)$ equivalently by the condition $\rho, i_0 \models \psi_f$ where 

$$\psi_f = \bigvee \limits_{\forall 1\le j\le k, q''_j, q'_j \in Q_j}
\left\{\begin{array}{c}\varphi''_1(init_1, q''_1) \wedge \varphi_1(q''_1, q'_2) \wedge \varphi''_2(q'_2, q''_2) \wedge \varphi_2(q''_2, q'_3) \ldots \wedge \varphi''_{k}(q'_k, q''_k)  \\ \wedge \\\varphi_k( q''_k, q'_{k+1}) \wedge \varphi''_{k+1}(q'_{k+1})
\end{array}\right\}
$$

where $Q_j$ is set of all the states of NFA $\re_j$, and $\varphi''_j$s and $\varphi_j$ are defined as above for $1\le j \le k$. Moreover, $\varphi''_{k+1}(q'_{k+1}) = \regm_{(u_k, \infty)} (\re_{k+1}[q'_{k+1}, F_{k+1}])$. Hence, we get the required reduction. Please refer to figure \ref{fig:fregtorat} for intuition.
\\(9) We now show that $\regm$ modality can also be expressed as $\fregmk^k$ modality for $k=3$. While this is not part of the main proof, this is an interesting observation which along with the reduction in the previous step (8) and applying structural induction on these input formulae proves the following theorem.
\begin{theorem}
    $\fregmk^k$ modality can be expressed as boolean combinations of $\regm$ modality, and vice-versa. Hence, both these modalities are expressively equivalent.
\end{theorem}
We again show the result for timing intervals of the form $[l,u]$, and the reduction can be generalized for other type of intervals. Notice that $\rho, i \models \regm){[l,u]}(\re)(S)$ iff (C1) either there exists a point $i_1''$ which is the last within interval $[0,l)$ from $i$, and there exists a point $i_2'$ which is the last point within interval $[l,uj]$ from $i$ such that $\Seg^+(\rho, i_1''+1, i_2', S) \in L(\re)$ (C2) or there is no point within   $[0,l)$ from $i$, in which case $i+1$ ( the next point from $i$) is within interval $[l,u]$ from $i$, and there exists a point $i_2'$ which is the last point within interval $[l,uj]$ from $i$ such that $\Seg^+(\rho, i+1, i_2', S) \in L(\re)$. Let $X^*$ be an NFA over $S$ accepting the words over $S$. Let $\re''(\psi')$ be the right quotient of $\re$ with respect to $\psi'$. Let $\re'(\varphi', \psi')$ be the left quotient of $\re''(\psi')$ with respect to $\varphi'$.

Notice that (C1) is equivalent to 
$$\bigvee \limits_{\varphi',\psi'}\fregmk^3_{[0,l),[l,u],[l,u](u,\infty)}(X^*.\varphi',\re',\psi',X^*)$$

Similarly (C2) is equivalent to 
$$\Box_{[0,l)}(false)\wedge \bigvee \limits_{\psi' \in S }\fregmk^3_{[l,u],[l,u](u,\infty)}(\re'',\psi',X^*)$$

\begin{figure}
    \centering
    \includegraphics[width=1\linewidth]{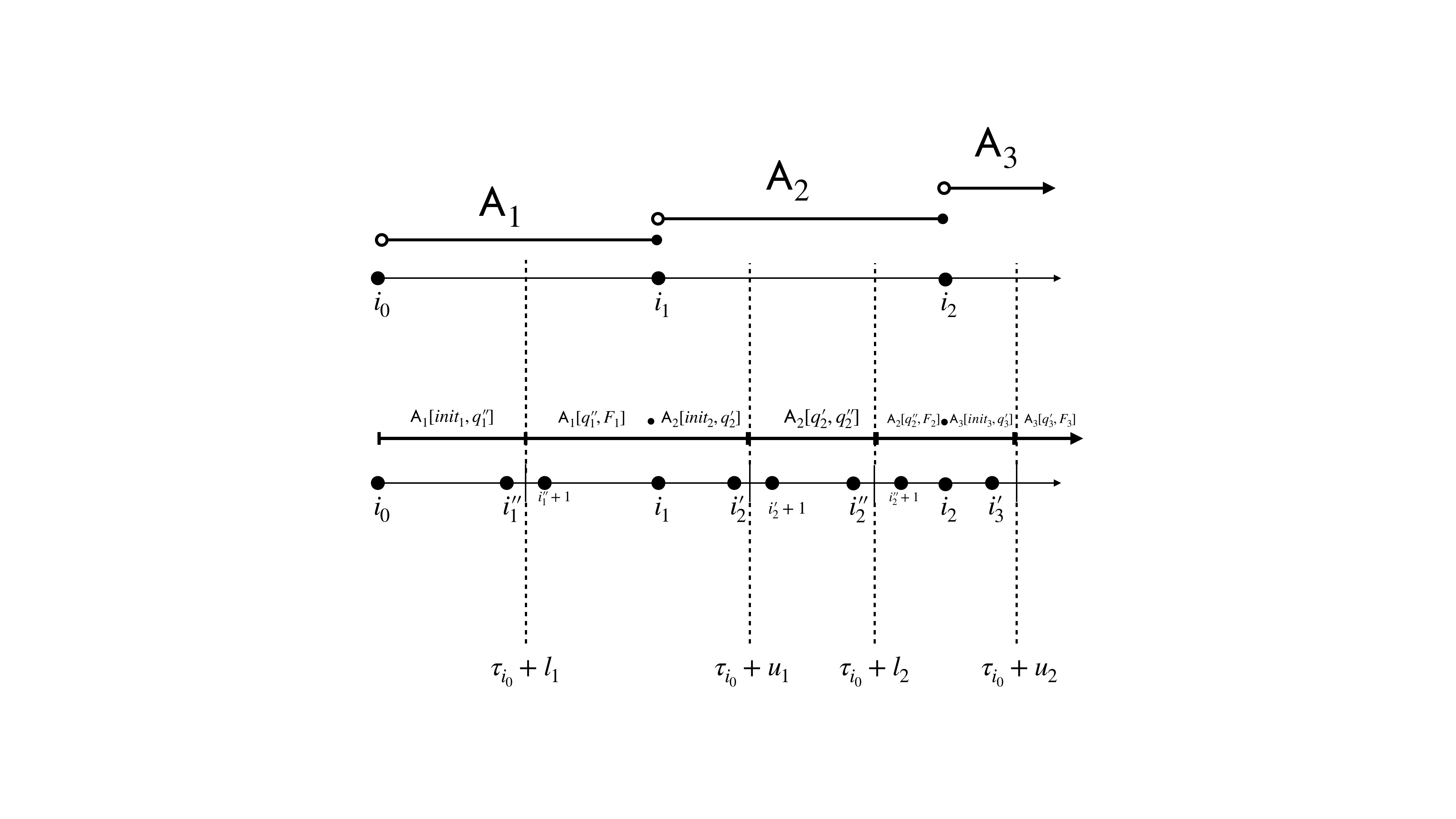}
    \caption{Figure showing the reduction of $\fregmk^2_{[l_1,u_1],[l_2,u_2]}(\re_1, \re_2, \re_3)(S)$ to formulae using only $\regm_I$ modalities.}
    \label{fig:fregtorat}
\end{figure}

\subsection{Expressiveness of $\patptl$}
In this section we briefly discuss the expressiveness of $\patptl$ vis-\`{a}-vis PMTL in the following theorem.
\begin{theorem}
    Partially Adjacent 1-$\tptl$ is strictly more expressive than Partially Punctual Metric Temporal Logic.
\end{theorem}
\begin{proof}
    Notice that any PMTL formula can be easily expressed in $\patptl$ formula using standard reduction from $\mtl$ to 1-$\tptl$. Moreover, by \cite{pandyasimoni} and \cite{rabinovichY}, $\mtl$ (and hence PMTL) can not express the following property which is clearly a non-adjacent (and hence also a partially adjacent 1-$\tptl$ formula).
$$\psi = \fut [a\wedge x.\{\fut (b \wedge T-x \in (1,2) \wedge \fut( c\wedge T-x \in (1,2)))\}]$$
\end{proof}
\begin{theorem}
    Partially Adjacent 1-$\tptl$ is strictly more expressive than Non-Adjacent 1-$\tptl$.
\end{theorem}
\begin{proof}
    This is a straightforward consequence of the latter not being able to express punctual constraints, while former can still express punctual constraints in either future or past but not both directions.
\end{proof}

This makes $\patptl$ the most expressive decidable logic over finite timed words known in the literature to the best of our knowledge.

\section{Conclusion}
Inspired by the notion of non-punctuality of \cite{AFH96}\cite{AlurFH91} and partial punctuality of \cite{time14}, we propose a notion of openness and partial adjacency, respectively. Openness is a stricter restriction as compared to non-punctuality. We show that openness does not make the satisfiability checking of 1-$\tptl$ (and $\tptl$) computationally easier, proving that $1$-$\tptl$ (and $\tptl$) does not enjoy the benefits of relaxing punctuality unlike $\mtl$. This makes a strong case for the notion of non-adjacency proposed for $1$-$\tptl$ by \cite{KKMP21}\cite{KKMP23}. We then propose a notion of partial adjacency, generalizing non-adjacency of \cite{KKMP21}. We study the fragment of 1-$\tptl$ with partial adjacency. This fragment allows adjacency in either future direction or past direction, but not both. We show that satisfiability checking for this fragment is decidable over finite timed words. The non-primitive recursive hardness for satisfiability checking of this fragment is inherited by its subclass 1-$\tptl[\until]$ and $\mtl[\until_I]$. We show that this fragment is strictly more expressive than Partially Punctual Metric Temporal Logic of \cite{time14}. This makes this logic one of the most expressive boolean closed real-time logic for which satisfiability checking is decidable over finite timed words, to the best of our knowledge. In the future, we plan to study the first-order fragment corresponding to this logic.

\newpage

\bibliographystyle{plain}
\bibliography{papers}

\newpage




\appendix
\section {Equisatisfiability Modulo Projections}
We denote the set of all the finite timed words over $\Sigma$ is denoted by $\tap \Sigma^*$.
\subsection{Simple Projections}
A  simple-$(\Sigma,X)$-behaviour is a timed word $\rho' \in \tap (\Sigma \cup X)^*$ such that for all $i \in dom(\rho')$ $\rho',i\models \bigvee \Sigma$. In other words, at any point $i \in dom(\rho')$, $\rho'[i]=(\sigma'_i,\tau'_i)$ and $\sigma'_i \cap \Sigma \ne \emptyset$. 

\noindent{\it \underline {Simple Projections}}: Consider a simple-$(\Sigma,X)$-behaviour $\rho'$. For $i \in dom(\rho')$ let $\rho'[i] = (\sigma'_i,\tau'_i)$. We define the {\it simple projection} of $\rho'$ with respect to $X$, denoted $\rho' \downarrow X$, as the word $\rho$, such that $dom(\rho) = dom(\rho')$ and for all $i \in dom(\rho)$ $\rho[i] = (\sigma'_i \setminus X, \tau'_i)$.

As an example, for $\Sigma=\{a,b\}, X=\{c,d\}$,  
$\rho' = (\{a,d\},0)(\{b,c\},0.3)(\{a,b,d\},1.1)$ is a simple-$(\Sigma,X)$-behaviour, while $\rho''=(\{a\},0)(\{c,d\},0.3)(\{b,d\},1.1)$ is not as, for  $i=2$, $\{c,d\} \cap \Sigma = \emptyset$. Moreover, $\rho' \setminus X = \rho =  (\{a\},0)(\{b\},0.3)(\{a,b\},1.1)$.\\
\noindent{\it{\underline{Equisatisfiability modulo Simple Projections}}}: 
Given any temporal logic formulae $\phi$ and $\psi$, let $\phi$ be interpreted over $\tap \Sigma^*$ and $\psi$ over $\tap (\Sigma\cup X)^*$. We say that $\phi$ is equisatisfiable to $\psi$ 
{\it modulo simple projections} iff $\Sigma,X$ are disjoint sets and,  
\begin{enumerate}
\item For any timed word $\rho' \in \tap (\Sigma \cup X)^*$, 
     $\rho' \models \psi$ implies that $\rho'$ is a simple-$(\Sigma,X)$-behaviour and $\rho' \downarrow X \models \phi,$
\item For any timed word $\rho\in \tap \Sigma^*$, $\rho \models \phi$ implies that there exists a timed word $\rho' \in \tap (\Sigma \cup X)^*$ such that
 $\rho' \downarrow X = \rho$ and $\rho' \models \psi$.
\end{enumerate}
We denote by $\psi \downarrow X \equiv_{\Sigma} \phi$, 
the fact that $\phi$ is equisatisfiable to $\psi$ modulo simple projections. 
Intuitively, $\psi \downarrow X \equiv_\Sigma \phi$ states that every model of $\phi$ is represented by some extended model of $\psi$ and every model
of $\psi$, when projected over the original alphabets, gives a model of $\phi$. Thus we get the following proposition.
\begin{proposition}
\label{simple-prop-1}
If $\phi$ is interpreted over $\tap \Sigma^*$, $\psi$ over $\tap(\Sigma\cup X)^*$ and $\psi \downarrow X \equiv_\Sigma \phi$ then $\phi$ is satisfiable if and only if $\psi$ is satisfiable.
\end{proposition}



Consider a formula $\psi = \wB(\bigvee X \rightarrow \bigvee \Sigma)$. If $\psi$ is interpreted over $\tap (\Sigma \cup X)^*$ then all the words satisfying $\psi$  are simple-$(\Sigma, X)$-behaviour. But the same $\psi$ is interpreted over $\tap (\Sigma, X')^*$ where $X' \supset X$, then $\psi$ will allow models where at certain points $\neg \bigvee \Sigma$ is true. 
Consider yet another formula $\phi = \fut_{(0,1)}(\fut_{(0,1)} a) \wedge \neg \fut_{(0,1)} a$. If $\phi$ is interpreted over $\tap \Sigma^*$ where $\Sigma= \{a\}$, then $\phi$ is unsatisfiable as there are no point with timestamp in $(0,1)$ (because of $\neg \fut_{(0,1)} a$) where 
subformulae $\fut_{(0,1)} a$ at depth 2 could have been asserted). But the same formulae when asserted over $\tap (\{a,b\})^*$ has a model $\rho = (a,0), (b,0.5), (a,1.1)$. Thus the satisfaction of a formulae is sensitive to the set of models it is interpreted over.

One way to make sure that the satisfaction of any formulae $\phi$ remains insensitive to the set of models it is interpreted over, is to restrict the language of the formulae, $\phi$ to contain only simple-$(\Sigma,X)$-behaviours for any $X$ irrespective of what propositions it is interpreted over. This could be done by conjuncting $\phi$ with $\wB(\bigvee \Sigma)$. 
To be more precise, the following proposition holds:
\begin{proposition}
\label{simple-prop-2}
For any $X$ such that $X$ and $\Sigma$ are disjoint set of propositions, if $\phi$ is interpreted over $\tap \Sigma^*$ and $\phi \wedge \wB(\bigvee \Sigma)$ over $\tap (\Sigma \cup X)^*$ then $L(\phi \wedge \wB(\bigvee \Sigma)) = \{\rho' | \rho' \downarrow X \models \phi\}$. 
\end{proposition}

 

\begin{proposition}
\label{simple-prop-3}
Let $\phi_1$ and $\phi_2$ be interpreted over $\tap \Sigma^*$. If $\psi_1 \downarrow X_1 \equiv_\Sigma  \varphi_1$ (i) and $\psi_2 \downarrow X_2 \equiv_\Sigma \varphi_2$(ii) and $X_1 \cap X_2 = \emptyset$ then $[\psi_1 \wedge \psi_2 \wedge \wB(\bigvee \Sigma)] \downarrow (X_1 \cup X_2) \equiv_\Sigma \varphi_1 \wedge \varphi_2$.
\end{proposition}
\begin{proof}
Consider any word $\rho \in \tap \Sigma^*$ such that $\rho \models \varphi_1 \wedge \varphi_2$. We show that there exists a simple-$(\Sigma,X_1\cup X_2)$-behaviour $\rho'$ such  that $\rho'$ models $\psi_1\wedge\psi_2 \wedge\wB(\bigvee \Sigma)$ and $\rho' \downarrow X_1\cup X_2 = \rho$.
Then by (i) there exists a simple-$(\Sigma,X_1)$-behaviour $\rho_1$ such that $\rho_1 \models \psi_1$ and $\rho_1 \downarrow X_1 =\rho$. Similarly, by (ii) there exists a simple-$(\Sigma,X_2)$-behaviour $\rho_2$ such that $\rho_2 \models \psi_2$ and $\rho_2 \downarrow X_2 = \rho$. Note that by proposition \ref{simple-prop-2}
all the words in $R_1 = \{\rho'_1 | \rho'_1 \downarrow X_2 = \rho_1\}$ satisfy $R_1 \wedge \wB(\bigvee \Sigma)$. Similarly, all the words in $R_2 = \{\rho'_1 | \rho'_1 \downarrow X_2 = \rho_1\}$ satisfy $\psi_2 \wedge \wB(\bigvee \Sigma)$. Note that for any word $\rho' \in R_1 \cap R_2$ 
$\rho' \downarrow (X_1 \cup X_2) = \rho$. Also note that the set $R_1 \cap R_2$ is non-empty.\footnote{Let $\rho'_1 = (\sigma'1,\tau)$, $\rho'_2 = (\sigma'2, \tau)$ and $\rho = (\sigma,\tau)$. Consider the word $\rho' = (\sigma',\tau)$ over $\tap(\Sigma \cup X_1 \cup X_2)^*$ such that $dom(\rho) = dom(\rho')$ and for any $i \in dom(\rho)$, $\sigma'_i = \sigma'1_i \cup \sigma'2_i$. Then $\rho' \in R_1 \cap R_2$}. Hence for every word $\rho$ if $\rho \models \varphi_1 \wedge \varphi_2$ then there exists a $\rho'$ such that $\rho'$ is a simple-$(\Sigma, (X_1,\cup X_2))$-behaviour, $\rho' \models \psi_1 \wedge \psi_2 \wedge \wB(\bigvee \Sigma)$ and $\rho' \downarrow (X_1 \cup X_2) = \rho$. The other direction can be proved similarly.
\end{proof}
\subsubsection{Flattening: An Example for Simple Projections.}

\label{sec-flat}
Let $\varphi \in \mtlpwuisi$ be a formula interpreted over  $\tap \Sigma^*$. 
Given any sub-formula $\psi_i$ of $\varphi$, and a fresh symbol $b_i \notin \Sigma$, $T_i=\wB(\psi_i \leftrightarrow b_i)$ is called a {\it temporal definition} and  $b_i$ is called a {\it witness}.  Let $\psi=\varphi[b_i /\psi_i]$ be the formula obtained by replacing all occurrences of $\psi_i$ in $\varphi$, with the witness $b_i$.    
This process of flattening is done recursively until we have replaced all future/past modalities of interest with witness variables, obtaining $\varphi_{flat}=\psi \wedge T \wedge \wB(\bigvee \Sigma)$, where $T$ is the conjunction of all temporal definitions. Note that the conjunct $\wB(\bigvee \Sigma)$ makes sure that the behaviours in the language of the formulae $\varphi_{flat}$ is restricted to be simple-$(\Sigma, X)$- behaviour, for any $X$ such that $X$ and $\Sigma$ are disjoint set of propositions.

For example, consider the formula $\varphi=a \until_{[0,3]}(c \since(\past_{[0,1]}d))$. 
Replacing the $\since, \past$ modalities with witness propositions 
$w_1$ and $w_2$ we get $\psi=a \until_{[0,3]}w_1$, 
along with the temporal definitions 
$T_1=\wB(w_1 \leftrightarrow ((c \since w_2)))$ and  
$T_2=\wB (w_2 \leftrightarrow 
( \past_{[0,1]}d))$. Hence, $\varphi_{flat}=\psi \wedge T_1 \wedge T_2 \wedge \wB(\bigvee \Sigma)$ is obtained by flattening the $\since,\past$ modalities from $\varphi$.  
Here $W=\{w_1,w_2\}$.  
Given  a timed word $\rho$ over atomic propositions in $\Sigma$ and a formula $\varphi$ built from $\Sigma$, flattening results in a formula $\varphi_{flat}$ built from $\Sigma \cup W$ such that for any word $\rho'\in \tap (\Sigma \cup W)^*$,
$\rho' \models \varphi_{flat}$ iff  $\rho \models \varphi$.   
Hence, we have  $ \varphi_{flat} \downarrow W \equiv_\Sigma \varphi$.

\subsection{Oversampled Projections}
\label{app:oversampling}
An oversampled-$(\Sigma,X, OVS)$-behavior is a timed word $\rho'$ over $X \cup \Sigma$ such that $\bigvee \Sigma$ is true at the first point of $\rho'$. 
Given an oversampled -$(\Sigma,X)$-behaviour $\rho'$, any point $i \in dom(\rho')$ is called an oversampling point or a non-action point with respect to $\Sigma$, if and only if $\rho',i \nvDash \bigvee \Sigma$.
Similarly, all the points $i$ where $\bigvee \Sigma$ is true are called $\Sigma$ action points or just action points.

\noindent{\it \underline{Oversampled Projections}}:      
Given an oversampled-$(\Sigma,X)$-behavior $\rho'=(\sigma',\tau')$, we define the {\it oversampled projection} of $\rho'$ with respect to $X$, denoted $\rho' \Downarrow X$ as the timed word obtained 
by deleting all the oversampling points, and then erasing the symbols of $X$ from the remaining points  (i.e., all points $j$ such that $\sigma'_j \cap \Sigma \neq \emptyset$). 
Note that the result of oversampled projection $\rho$=$\rho'\Downarrow X$ is a timed word over propositions in $\Sigma$. 
\\As an example, let $\Sigma=\{a,b\}, X=\{c,d\}$. Consider, 
\\$\rho' = (\{a\},\underline{(\{c,d\},0.3)}(\{a,b\},0.7)(\{b,d\},1.1)$. $\rho'$ is an oversampled-$(\Sigma,X)$-behaviour where the underlined point is an oversampling point \footnote{For brevity, when we consider a behaviour to be an oversampled $(\Sigma,X)$-behaviour, we will call any point $i$ an oversampling point, iff time point $i$ is an oversampling point with respect to $\Sigma$.}. $\rho = (\{a\},0)(\{a,b\},0.7)(\{b\},1.1)$ is an oversampled projection of $\rho'$. Timed word $\delta = (\{c,d\},0)(\{a\},0.2)(\{a,b\},0.7)(\{b,d\},1.1)$ is not an oversampled-$(\Sigma,X)$-behavior, since at the first point,  no proposition in $\Sigma$ is true. 
\medskip
\\
\noindent{\it \underline{Equisatisfiability modulo Oversampled Projections}}:
Given any logic formulae $\phi$ and $\psi$, let $\phi$ be interpreted over $\tap \Sigma^*$ and $\psi$ over $\tap (\Sigma \cup X)^*$. We say that $\phi$ is equisatisfiable to $\psi$ 
{\it modulo oversampled projections} if and only if  $\Sigma \cap X$ is $\emptyset$,: 
\begin{enumerate}
\label{def}
\item For any $(\Sigma,X)$-oversampled behaviour $\rho'$,  $\rho' \models \psi$ implies that $\rho' \Downarrow X  \models \phi$ 
\item For any timed word $\rho \in \tap \Sigma^*$ such that  $\rho \models \phi$,  there exists $\rho'$ such that $\rho'$ is an oversampled-$(\Sigma,X)$-behaviour, $\rho' \models \psi$ and $\rho'\Downarrow X = \rho$. 
\end{enumerate}
 We denote by $\psi \Downarrow X \equiv_\Sigma \phi$, the fact that $\phi$ is equisatisfiable to $\psi$ modulo  oversampled projections. Similar to proposition \ref{simple-prop-1}, we get the following proposition:
\begin{proposition}
\label{oversampledpp1}
If $\phi$ is interpreted over $\tap \Sigma^*$, $\psi$ over $\tap(\Sigma\cup X)^*$ and $\psi \Downarrow X \equiv_\Sigma \phi$ then $\phi$ is satisfiable if and only if $\psi$ is satisfiable.
\end{proposition}
\subsubsection{Relativization: An Example of Oversampled Projections.}

Recall that if formulae $\phi = \fut_{(0,1)}(\fut_{(0,1)} a) \wedge \neg \fut_{(0,1)}$ is interpreted over $\tap \{a\}^*$ then it is unsatisfiable. While the same formulae when interpreted over $\tap (\{a,b\})^*$ has a model $\rho = (a,0), (b,0.5), (a,1.1)$. 

One way to make sure that the satisfaction of a formulae $\phi$ remains insensitive to the set of propositions over which it is interpreted is, by allowing models where all the points satisfy $\bigvee \Sigma$ as we did in \ref{simple-prop-2}. But this is too restrictive for our purposes, as the resulting formulae will be satisfied only by models which are simple-$(\Sigma,X)$ behaviours. Thus, we need a way of transforming a formulae $\phi$ to some $\phi'$ such that $\phi$ and $\phi'$ are equivalent when interpreted over $\tap \Sigma^*$. Moreover, $L(\phi') = \{\rho' | \rho' \Downarrow X \models \psi\}$. 

Note that, given $\Sigma$, for any $X$, all oversampled points in any oversampled-$(\Sigma, X)$-behaviours are characterized by the formula $\neg act$ where 
$act = \bigvee \Sigma$. To realize such a transformation, the trick is to relativize the original formula to points where $act$ is true. That is, we assert inductively, that a subformula is evaluated for its truth only at action points and the oversampling points are neglected. 

\oomit{
Consider a formula $\phi = (a \until b)\wedge \fut_{(1,2)} c$. Now we take the subformula $\fut_{(1,2)}c$ and change it to an equisatisfiable formulae $\fut_{[1,1]} int_1 \wedge \neg \fut_(0,1) int_1 \wedge \neg \fut_{(1,\infty)} int_1 \wedge \fut (int\wedge \fut_{(0,1)} c$. This formulae adds an oversampling point with proposition $int_1$ at $time = 1$ and then it asserts $\fut_{(0,1)} c$ from that point. But notice that this oversampling might not allow the formulae $(a \until b)$ to be true if the first $b$ occurs after 1 time unit. But on the original model(non oversampled model) this subformula may be true. Thus while oversampling we need to assert that the other subformula ignore these extra action points when evaluated for the truth.

Relativization is a technique to ignore oversampling points while evaluating truth of a formula (or a subformula) thus restricting the subformula to be evaluated only at action points.
}
Let $\psi$ be a formula interpreted on timed words over atomic propositions from $\Sigma \cup X$.  
The {\it Relativization} with respect to $\Sigma$ of $\psi$ denoted $Rel(\Sigma,\psi)$ is obtained by replacing recursively 
\begin{itemize}
\item All subformulae 
of the form $\phi_i \until_I \phi_j$  with  \\ $(act \rightarrow Rel(\Sigma,\phi_i)) \until_I (Rel(\Sigma,\phi_j)\wedge act)$,
\item All subformulae of the form  $\phi_i \since_I \phi_j$ with \\  $(act \rightarrow Rel(\Sigma,\phi_i)) \since_I (Rel(\Sigma,\phi_j)\wedge act)$.
\item All subformulae of the form $\Box_I \phi$ with \\ $\Box_I(act \rightarrow Rel(\Sigma, \phi))$, and all subformulae of the form $\fut_I \phi$ with $\fut_I(Rel(\Sigma, \phi) \wedge act)$.
\end{itemize}  
This allows  relativization with respect to the propositions in $\Sigma$.
Let  $\psi{=}\varphi_1 \until_I (\varphi_2 \wedge \Box \varphi_3)$, 
and $\zeta_i{=}Rel(\Sigma,\varphi_i)$ for $i{=}1,2,3$. 
Then $Rel(\Sigma,\psi){=}(act {\rightarrow} \zeta_1) \until_I (act \wedge [\zeta_2 \wedge \Box (act {\rightarrow} \zeta_3)])$.

\end{document}